\documentclass[letterpaper,twocolumn,eqsecnum,nosuperscriptaddress,nofootinbib]{revtex4}

\usepackage{amssymb}
\usepackage{amsmath}
\usepackage{amsthm}
\usepackage{bbm}
\usepackage{hyperref}
\usepackage{graphicx}
\usepackage[figure]{hypcap}

\usepackage{chngcntr}
\counterwithout{equation}{section}

\newcounter{ifsup}

\theoremstyle{remark}
\newtheorem{theorem}{Theorem}

\newtheorem{athm}{Theorem}[section]
\newtheorem{lemma}[athm]{Lemma}
\newtheorem{prop}[athm]{Proposition}
\newtheorem{coro}[athm]{Corollary}
\newtheorem{obs}[athm]{Observation}
\newtheorem{prty}[athm]{Property}

\newtheorem*{definition}{Definition}
\newtheorem{conj}{Conjecture}

\newcommand*{\eins}{\ensuremath{\mathbbm 1}}

\newcommand*{\bbC}{\mathbb{C}}

\newcommand*{\cA}{\mathrm{A}}
\newcommand*{\cE}{\mathcal{E}}
\newcommand*{\cR}{\mathcal{R}}
\newcommand*{\cS}{\mathrm{S}}

\newcommand*{\ket}[1]{\left|#1\right\rangle}
\newcommand*{\bra}[1]{\left\langle #1\right|}
\newcommand*{\proj}[1]{\ket{#1}\bra{#1}}
\newcommand*{\Tr}{\mathsf{Tr}}
\newcommand*{\fr}[2]{\frac{#1}{#2}}

\newcommand{\nrm}[1]{\left\|#1\right\|}
\newcommand{\abs}[1]{\left|#1\right|}
\newcommand{\vc}[1]{\boldsymbol{#1}}
\newcommand\dsm{\mathrel{\stackrel{\makebox[0pt]{\mbox{\normalfont\tiny UT}}}{\succ}}}
\newcommand\thm{\mathrel{\stackrel{\makebox[0pt]{\mbox{\normalfont\tiny Th}}}{\succ}}}
\newcommand{\be}{\begin{equation}}
\newcommand{\ee}{\end{equation}}

\hypersetup{
    pdftoolbar=true,        % show Acrobat’s toolbar?
    pdfmenubar=true,        % show Acrobat’s menu?
    pdffitwindow=false,     % window fit to page when opened
    pdfstartview={FitH},    % fits the width of the page to the window
}

\renewcommand\thesection{\arabic{section}}

\begin{document}
\title{Low-temperature thermodynamics with quantum coherence}
\date{\today}
\author{Varun Narasimhachar}
\email{vnarasim@ucalgary.ca}
\author{Gilad Gour}
\address{Department of Mathematics and Statistics and Institute for Quantum Science and Technology, University of Calgary, 2500 University Drive NW, Calgary, Alberta, Canada T2N 1N4}

\begin{abstract}
Thermal operations are an operational model of non-equilibrium quantum thermodynamics. In the absence of coherence between energy levels, exact state transition conditions under thermal operations are known in terms of a mathematical relation called thermo-majorization. But incorporating coherence has turned out to be challenging, even under the relatively tractable model wherein all Gibbs state-preserving quantum channels are included. Here we find a mathematical generalization of thermal operations at low temperatures, ``cooling maps'', for which we derive the necessary and sufficient state transition condition. Cooling maps that saturate recently-discovered bounds on coherence transfer are realizable as thermal operations, motivating us to conjecture that all cooling maps are thermal operations. Cooling maps, though a less conservative generalization to thermal operations, are more tractable than Gibbs-preserving operations, suggesting that cooling map-like models at general temperatures could be of use in gaining insight about thermal operations.
\end{abstract}

\maketitle
\section{Introduction}
Advancements in cryogenics have enabled us to prepare systems at very low temperatures using various cooling techniques \cite{Refr,Ncoo,Fcoo}. In fact, humans may soon cool systems to levels that are not known to exist anywhere in the observable universe! Low-temperature systems exhibit exotic, characteristically quantum phenomena such as the quantum hall effect, superconductivity, and topological order \cite{QHE,Scon,Tord}, enabling diverse technological applications such as precision measurement instruments \cite{KK03,thme}, fast digital electronics \cite{DS10}, and NMR applications \cite{NMR}. One of the biggest potential applications is quantum computing: several of the proposed implementations of quantum computing are currently dependent on low-temperature capability \cite{QC2,QC,QC3,Flqu}. In addition, low-temperature systems are useful in fundamental research frontiers such as particle physics \cite{Acc} and dark matter detection \cite{DM}.

The prevalence of such phenomena at low temperatures is related to the fact that coherence can better endure thermal noise at low temperatures \cite{CoCoh}. On the other hand, significant strides have been made in realizing coherent quantum phenomena at higher temperatures \cite{highTc}. These developments mean that more and more experimentally realizable systems exhibit effectively ``low-temperature-like'' behavior at temperatures that are no longer forbiddingly low.

Our ability to control and manipulate physical systems in either of these cases---actual or effective low-temperature settings---hinges on our understanding of the thermodynamics of low-temperature environments. While classical thermodynamics is an adequate tool for analyzing macroscopic systems in thermodynamic equilibrium, it proves inadequate in any situation involving microscopic quantum systems or thermodynamic non-equilibrium. There has been extensive interest in formulating a theory of thermodynamics applicable to such situations \cite{janz,smarf,nege,qlan,Aab,Nan,Reth,Sec,NU,catcoh,SSP14,CohRE,MJR14,CorW,Coh,CST}. Most of these works, especially the recent ones, have studied thermodynamical processes using a model called ``thermal operations'', which are defined operationally as processes realizable by coupling a system with a heat reservoir and carrying out a global energy-conserving unitary evolution. However, existing formulations have not been able to fully incorporate quantum coherence---the essential aspect of quantum physics that is represented in the iconic ``Schr\"odinger's cat'' thought experiment. While coherence becomes irrelevant in the special case where the Hamiltonian of a system is fully degenerate \cite{NU}, it is essential to understanding the thermodynamics of general systems. Moreover, coherence is a resource, helpful both in thermodynamic tasks such as work extraction \cite{catcoh,CorW,SSP14} and in other resource-based tasks such as reference frame alignment \cite{BRS07}. A recent surge of work in the field has made progress in understanding the role of coherence in thermodynamics \cite{CohRE,MJR14,CorW,Coh,CST}.

In \cite{Coh}, the authors find an upper bound to the extent to which coherence can be preserved under thermal operations. In \cite{CST}, progress beyond such bounds has been made, but exact state transition conditions remain elusive. A possible strategy to gain further insight is to consider a set of processes beyond thermal operations, namely all quantum channels that preserve the Gibbs state. However, it is not clear if this expanded set is physically motivated, because it is defined mathematically, rather than operationally.

In this paper we report a mathematical generalization of coherent thermal operations at low temperatures, that we call the ``cooling maps'' model. At temperatures low enough for the ambient bath to be approximately in its ground state, thermal operations have the effect of taking away heat from the system of interest, therefore cooling the system. This motivates our definition of ``cooling maps''. We find the necessary and sufficient condition for state transitions to be feasible under these maps. We construct thermal operation implementations for cooling maps that saturate the coherence transfer bounds of Ref.~\cite{Coh}, opening up the possibility of improvements in coherence-based tasks. Our work also sheds light on the relations between different models that could be used to study low-temperature quantum thermodynamics: Thermal operations, Gibbs-preserving operations, and our cooling maps (see Fig.~\ref{fhier}). In general, cooling maps are potentially less conservative than thermal operations, in the sense that processes that are forbidden under thermal operations could be allowed under cooling maps. However, we demonstrate that the latter are much more conservative than Gibbs-preserving operations. Although the cooling maps model emerges from the low-temperature limit, the methods used in our work could potentially lead to a better understanding of coherence in quantum thermodynamics at all temperatures, in conjunction with the methods and results from other recent works that address this subject.
\begin{figure}%[t]
    \includegraphics[width=.9\columnwidth]{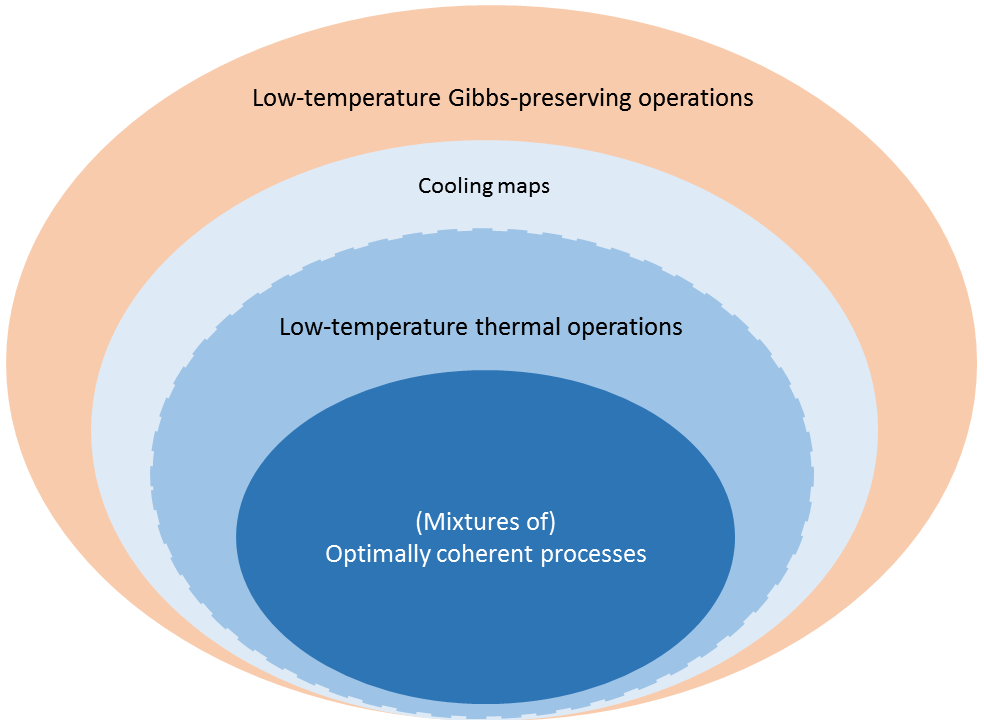}
    \caption[justification=justified]{[Inclusion hierarchy of thermodynamic models] In this work we introduce the cooling maps as a generalization of low-temperature thermal operations, and the dashed boundary between the two sets indicates that the sets of state transitions they admit might coincide. Thermal operations include cooling maps that optimally preserve coherences.}
    \label{fhier}
\end{figure}
\section{Results}

\subsection*{Background: thermal operations}
The physical setting in our model is a $d$-dimensional quantum system $\cS$ whose free Hamiltonian is $H_\cS$. For convenience, we make some simplifying assumptions on $H_\cS$. Firstly, that $H_\cS$ has no degenerate energy levels. Thus, its energy spectrum has the structure
\be E_1<E_2<\dots<E_d.\ee
We also assume that $E_i-E_j\ne E_k-E_l$ for any two pairs of indices $(i,j)$ and $(k,l)$, except when either $i=j$ and $k=l$, or $i=k$ and $j=l$.

Note that these assumptions are satisfied for almost all Hamiltonians, in a statistical sense: the subset of matrices that fail to satisfy these assumptions is of measure zero in the set of all Hermitian matrices. One may dismiss this measure-theoretic argument on the grounds that Hamiltonians of typical naturally occurring systems, such as atoms, have degenerate levels and gaps. But these degeneracies can be broken with the slightest perturbation, such as an external electromagnetic field. The absence of any such perturbation is in fact an exceptional circumstance, and it is reasonable to suppose that the above assumptions are satisfied by most realistic physical systems. Moreover, certain physical systems that are used in applications have these properties. For example, different types of superconducting artificial atoms used in quantum computing implementations, such as Cooper-pair boxes and transmons, are governed by anharmonic-oscillator{\textendash}like Hamiltonians \cite{PhQd}. Another important point to consider about these assumptions is the scope of their impact on our results. For systems that do satisfy these assumptions, the state transition conditions that we will derive turn out to be necessary and sufficient. However, even for systems that fail to satisfy these assumptions, our conditions remain sufficient, only losing their necessity. Furthermore, our results on maximally coherent processes hold regardless of these assumptions.

The non-degeneracy of all energy levels of $H_\cS$ implies that we can label the eigenvectors (stationary states) using just one label, as in $\ket{E_j}$. If $\cS$ is isolated, its dynamics is governed by the Schr\"odinger equation under $H_\cS$. If, instead, it is capable of exchanging heat with a thermal reservoir (heat bath) at temperature $T$, then $\cS$ eventually ``equilibrates'', i.e. approaches the state of thermal equilibrium with the reservoir, regardless of its initial state. The equilibrium state is given by the so-called Gibbs state
\be\gamma_\cS:=\fr1{Z_\cS}\exp\left(-\beta H_\cS\right)=\sum_{j=1}^d\fr{\exp\left(-\beta E_j\right)}{Z_\cS}\proj{E_j},\ee
where $\beta=(k_BT)^{-1}$ with $k_B$ the Boltzmann constant, and $Z_\cS:=\sum_{j=1}^d\exp\left(-\beta E_j\right)$ the partition function of $\cS$.

Quantum thermodynamics enables us to go beyond just this asymptotic description and to determine what processes can occur in the course of equilibration. If the bath is ``large'' enough, every possible physical process occurring on the system $\cS$ can be modeled through the following stepwise operational form:
\begin{enumerate}
\item Bring $\cS$ (which is initially isolated) together with an arbitrary ancillary system $\cA$, which is prepared in its own Gibbs state $\gamma_\cA:=\left(1/Z_\cA\right)\exp\left(-\beta H_\cA\right)$ corresponding to its own free Hamiltonian $H_\cA$ and the ambient temperature $T$. Physically, the ancilla is all or part of the heat bath.
\item Perform any global energy-conserving unitary evolution $U$ on the composite $\cS\cA$. Energy conservation is imposed through the commutator relation $[U,H_{\cS\cA}]:=0$, where $H_{\cS\cA}$ is the Hamiltonian that governs uncoupled evolution of the composite system $\cS\cA$:
\be H_{\cS\cA}:=H_\cS\otimes\eins_\cA+\eins_\cS\otimes H_\cA.\ee
\item Discard the ancilla $\cA$ (i.e., isolate $\cS$ again).
\end{enumerate}
Mathematically, the process is represented by a completely positive trace-preserving map $\cE$ whose action on an arbitrary state $\rho$ of $\cS$ is given by
\be\rho\mapsto\cE(\rho)=\Tr_\cA\left[U\left(\rho\otimes\gamma_\cA\right)U^\dagger\right],\ee
where $\Tr_\cA$ is the mathematical operation of partial trace with respect to $\cA$, corresponding to the physical operation of discarding the system $\cA$.

Processes modeled in this manner have been called thermal operations in the literature (see Supplementary Section~\ref{sther} for details). The energy conservation condition $[U,H_{\cS\cA}]:=0$ can be understood in terms of the eigenvalues and eigenvectors of $H_{\cS\cA}$: If $\{G_j\}$ are the eigenvalues of $H_{\cS\cA}$, and $\ket{G_j;\alpha}$ represents an eigenvector belonging to $G_j$ (where $\alpha$ could be a label identifying eigenstates within a degenerate energy level), then we require
\begin{equation}
\bra{G_j;\alpha}U\ket{G_k;\beta}=0\;\forall G_j\ne G_k.
\end{equation}
The uncoupled structure of $H_{\cS\cA}$ means that its energy levels have the form $G=E+F$, where $E$ and $F$ are eigenvalues of $H_\cS$ and $H_\cA$, respectively. A unitary such as $U$ can change the state of $\cS$ by raising (lowering) $E$ while simultaneously lowering (raising, respectively) $F$ so as to keep $G$ constant.

\subsection*{The emergence of ``cooling maps''}
When the ambient bath temperature is low enough, the initial state of any ancilla $\cA$ drawn from the bath (i.e., its Gibbs state) is almost entirely in its lowest energy level:
\begin{equation}\label{lowt}
\gamma_\cA\approx\fr1{g_1}\sum_{t=1}^{g_1}\proj{F_1;t},
\end{equation}
where $F_1$ is the ground state energy, $g_1$ the multiplicity of this energy level, and $t$ some label that identifies eigenvectors within the degenerate subspace. This means that even though the temperature is non-zero, the bath effectively behaves as though it were zero. Since the ancilla $\cA$ starts out in its lowest energy level, any energy transfer that an energy-conserving unitary $U$ causes between $\cS$ and $\cA$ must be from $\cS$ to $\cA$. Therefore, the effect of a low-temperature thermal operation on $\cS$ is to ``cool'' it. How low the temperature needs to be in order for this approximation to be valid is determined by the composition of the system and the bath (see Supplementary Section~\ref{slowt} for details). For example, a bath consisting of many identical systems in the same Gibbs state (i.e., of the form $\gamma^{\otimes n}$) would satisfy this approximation at temperatures much lower than the gap between the ground and first excitated state of each subsystem. In some cases one can infer this low-temperature behavior of the bath indirectly, through the behavior of the system of interest. For example, the ambient bath surrounding a superconducting artificial atom behaves effectively in this manner at temperatures lower than the superconducting critical temperature of the system. Under condition (\ref{lowt}), together with our assumption of non-degenerate energy levels and gaps in $H_\cS$, all thermal operations reduce to an elegant form, characterized by a Kraus operator sum representation with the following features: A number $n\le d$ of diagonal Kraus operators
\be K_i=\sum_{j=1}^d\lambda_j^{(i)}\proj{E_j},\ee
$i\in\{1\dots n\}$; and $d(d-1)/2$ Kraus operators of the form
\be J_{jk}=\mu_{jk}\ket{E_j}\bra{E_k},\ee
one for each pair $(j,k)$ with $j<k$. Note that some of the $J$'s could be zero. If we relax the non-degeneracy conditions on the system Hamiltonian, the form of these Kraus operators generalizes to the well known structure of amplitude-damping channels, which are used as a model of dissipation, spontaneous emission, etc. \cite{Flqu}. The $j<k$ condition in the $J_{jk}$'s captures the ``cooling'' action that results from the low-temperature assumption. This motivates us to call any process with such an operator sum representation a ``cooling maps''. A detailed derivation of this form may be found in Supplementary Section~\ref{scool}.

\subsection*{The action of cooling maps}
Let us denote by $\cE$ the channel realized by the above Kraus operators. The action of $\cE$ on the state of $\cS$ can be expressed succinctly if we group the $\lambda$'s into $d$ vectors of the form $\vc\lambda_j\equiv(\lambda_j^{(1)}\dots\lambda_j^{(n)})^T$. If $\rho$ is the initial state and $\sigma=\cE(\rho)$ the state after the application of $\cE$, then the relation between the off-diagonal elements of $\rho$ and $\sigma$ is simple:
\be\sigma_{jk}=\left\langle\vc\lambda_j,\vc\lambda_k\right\rangle\rho_{jk},\ee
for each $j\ne k$. Here $\langle\cdot,\cdot\rangle$ denotes the usual inner product between two vectors. On the other hand, the relation between the diagonal parts of the states is given by
\be\sigma_{jj}=\left\langle\vc\lambda_j,\vc\lambda_j\right\rangle\rho_{jj}+\sum_{k>j}\abs{\mu_{jk}}^2\rho_{kk}.\ee
The matrix $q$ whose components are the quantities $q_{jk}:=\left\langle\vc\lambda_j,\vc\lambda_k\right\rangle$ appearing above is called the Gramian of the collection $\{\vc\lambda_j\}$. Every Gramian matrix is positive-semidefinite, and conversely, every positive-semidefinite matrix is the Gramian of some collection of vectors \cite{HJ}.

If we view the diagonal $\vc u\equiv\left(\rho_{11}\dots\rho_{dd}\right)^T$ as a classical probability distribution, then its transformation under $\cE$ can be represented by the action of a matrix $P$:
\be\vc v\equiv\left(\sigma_{11}\dots\sigma_{dd}\right)^T=P\vc u,\ee
where the components of $P$ are given by
\be P_{j|k}=\left\{\begin{array}{ll}\left\langle\vc\lambda_j,\vc\lambda_j\right\rangle,&\textnormal{if }j=k;\\
\abs{\mu_{jk}}^2,&\textnormal{if }j\ne k.\end{array}\right.\ee
$P$ is upper-triangular: $P_{j|k}=0$ if $j>k$. Furthermore, it is column-stochastic: $P_{j|k}\ge0$  for all $(j,k)$; and $\sum_{j=1}^dP_{j|k}=1$ for all $k$.
\subsection*{Upper-triangular majorization}
In Supplementary Lemma~\ref{Pcon} we prove that the existence of an upper-triangular (UT) column-stochastic matrix $P$ such that $\vc v=P\vc u$ is in fact equivalent to the simultaneous fulfillment of the following $(d-1)$ inequalities:
\begin{align}
u_d&\ge v_d,\nonumber\\
u_{d-1}+u_d&\ge v_{d-1}+v_d,\nonumber\\
&\;\;\vdots\nonumber\\
u_2+u_3\dots+u_d&\ge v_2+v_3\dots+v_d.
\end{align}
We abbreviate the above inequalities collectively as $\vc u\dsm\vc v$, read ``$\vc u$ UT-majorizes $\vc v$''. In the literature, UT majorization has variously been referred to as ``unordered majorization'' \cite{SMOA} and ``majorization'' \cite{SUTM} (not to be confused with the more common established sense of the term ``majorization''), as well as the term we use \cite{SUTM2}. It is instructive to compare UT majorization with the so-called thermo-majorization, which governs the transformation of the diagonal elements in thermodynamics at general temperatures \cite{Nan}. The thermo-majorization relation between two probability distributions $\vc u$ and $\vc v$ can be defined in different ways, of which the following is perhaps most intuitive. Denote by $\vc u_\gamma$ the Gibbs distribution for the given Hamiltonian at some inverse temperature $\beta$. That is,
\begin{equation}
u_{\gamma,j}=\fr1{Z_\cS}e^{-\beta E_j}.
\end{equation}
Then we say that ``$\vc u$ thermo-majorizes $\vc v$'' if there exists a column-stochastic matrix $M$ such that $M\vc u_\gamma=\vc u_\gamma$, i.e., $M$ fixes the Gibbs distribution, and $M\vc u=\vc v$. Considering that UT stochastic matrices fix the zero-temperature limit of the Gibbs distribution for a non-degenerate Hamiltonian, it seems intuitively reasonable that UT majorization emerges as the zero-temperature limit of thermo-majorization. We show this rigorously in Supplementary Section~\ref{sreal}.

\subsection*{State transformation conditions}
The foregoing observations put together yield our main result: the necessary and sufficient condition for the feasibility of state transitions under cooling maps.
\begin{theorem}\label{thth}
Let $\rho$ and $\sigma$ be two states on $\cS$, arbitrary except that the matrix elements of $\rho$ are non-zero ($\rho_{jk}\ne0$). Define the matrix $Q$ as follows:
\be Q_{jk}=\left\{\begin{array}{ll}\min\left(\fr{\sigma_{jj}}{\rho_{jj}},1\right),&\textnormal{if }j=k;\\
\fr{\sigma_{jk}}{\rho_{jk}},&\textnormal{if }j\ne k.\end{array}\right.\ee
Then, the transition $\rho\mapsto\sigma$ is possible through a cooling map if and only if both the following conditions hold:
\begin{enumerate}
\item The diagonal parts $\vc u\equiv\left(\rho_{11}\dots\rho_{dd}\right)^T$ and $\vc v\equiv\left(\sigma_{11}\dots\sigma_{dd}\right)^T$ satisfy $\vc u\dsm\vc v$.
\item The matrix $Q$ is positive-semidefinite: $Q\ge0$.
\end{enumerate}
\end{theorem}
The $Q$ appearing above is in fact a special limiting case of the Gramian matrix $q$ that we introduced earlier. Note that we can easily adapt the theorem to cases where some of the $\rho_{jk}$'s are zero. Also note that if we relax the non-degeneracy assumptions on $H_\cS$, the condition of this Theorem remains sufficient for state transitions; it is, however, no longer necessary. We provide the proof of this theorem, as well as technical details of the preceding discussion, in Supplementary Section~\ref{stran}.

\subsection*{Optimally coherent cooling maps are thermal}
We constructed the cooling maps based on the low-temperature limit of thermal operations. Since the latter link the mathematical model with actual physics, we must determine if the cooling and low-temperature thermal models are equivalent, or if instead there exist state transitions achievable by cooling maps but forbidden under thermal operations. A couple of special cases support the equivalence hypothesis.

The first special case is when $\cS$ is a two-level system, i.e., $d=2$, for which cooling maps are identical with thermal operations. This can be proved simply by constructing a thermal implementation of any cooling map (Supplementary Corollary~\ref{c2l}). The state-transition conditions for two-level systems under thermal operations at any temperature have been derived recently by \'Cwikli\'nski \textit{et al.} \cite{Coh}, and our result tallies with theirs in the low-temperature limit.

The other special case involves pairs of states $(\rho,\sigma)$ satisfying the first condition of Theorem~\ref{thth} and also
\be Q_{jk}=\left(Q_{jj}Q_{kk}\right)^{1/2}\ee
for all $(j,k)$. Then there is a thermal operation taking $\rho\mapsto\sigma$ (Supplementary Corollary~\ref{opco}). The significance of this special case is that each off-diagonal element (i.e., coherence between different energy levels) in $\sigma$ has the highest magnitude possible, in the following sense. Suppose that $\sigma'$ is a state whose diagonal coincides with that of $\sigma$. Then, if $\rho\mapsto\sigma'$ is possible via a cooling map, then it holds for all $(j,k)$ that $\abs{\sigma'_{jk}}\le\abs{\sigma_{jk}}$ (Supplementary Corollary~\ref{opopco}). This bound was also proved for all temperatures in Ref.~\cite{Coh}, whose authors constructed examples where the bound cannot be attained. Our results show that it is always attainable at low temperatures. The same conclusion is reached in Ref.~\cite{CST}, where the high-temperature case is also considered. More generally, we prove that any mixture of optimally coherent processes is a low-temperature thermal operation (Supplementary Corollary~\ref{mixt}). In fact, this holds even if the non-degeneracy assumptions on the Hamiltonian $H_\cS$ are relaxed.

\subsection*{Gibbs-preserving operations}
In general, the set of cooling maps could be larger than that of thermal operations. Whether the two sets are equivalent is an open problem. There is, however, an even larger set that includes both of these: all processes $\cE$ that preserve the Gibbs state $\gamma_\cS$. That is, $\cE(\gamma_\cS)=\gamma_\cS$.
These processes, called the ``Gibbs-preserving operations'', have been studied in the past as a possible model for thermodynamic processes. Even if one favors thermal operations as the physically more reasonable model, the Gibbs-preserving model can be studied as an approximation to thermal operations that is potentially more mathematically tractable. It is not hard to verify that all cooling maps are Gibbs-preserving. We model the low-temperature limit for the Gibbs-preserving operations through the approximation $\gamma_\cS\approx\proj{E_1}$ (Supplementary Section~\ref{glowt}). This form follows from our non-degeneracy assumption on the system Hamiltonian $H_\cS$. Note that we do not require $\cS$ to actually be in the Gibbs state; we merely require the ambient temperature to be low enough for the Gibbs state to be approximately equal to the ground state. Considering this approximate form of the Gibbs state, low-temperature Gibbs-preserving operations are processes $\cE$ such that
\be\label{gpo}\cE\left(\proj{E_1}\right)=\proj{E_1}.\ee

\subsection*{Monotones under Gibbs-preserving operations}
Clearly, the structure Eq.~\eqref{gpo} of Gibbs-preserving operations privileges the $E_1$ energy level in relation to the rest of the state space, leading to the following canonical parametrization of a generic state:
\be\rho=\left(\begin{array}{c|c}
\alpha&\vc x^\dagger\\
\hline\vc x&A\end{array}\right),\ee
where $\alpha:=\bra{E_1}\rho\ket{E_1}\ge0$ is a real scalar, $\vc x$ is a complex $(d-1)$-dimensional vector, and $A$ is a $(d-1)$-dimensional subnormalized density operator. In fact, any $\rho$ can be reversibly converted (through an allowed unitary) to a state with a diagonal $A$ and nonnegative real entries in $\vc x$. The parameter $\alpha$ assumes its greatest value $1$ when $\rho$ coincides with the Gibbs state $\ket{E_1}$, and its least value $0$ when $\rho$ is supported on the subspace orthogonal to $\ket{E_1}$. Therefore, we can think of
\be\nu_\mathrm I(\rho):=1-\alpha\ee
as a measure of the deviation of $\rho$ from equilibrium, or in other words, its ``nonequilibrium'' (hence the letter $\nu$). However, this measure does not contain any information about the coherences between different energy levels: it measures the nonequilibrium manifest in the diagonal part of $\rho$, related to the statistical distribution of energy amongst different energy levels. This aspect of nonequilibrium has in the past been referred to as ``informational nonequilibrium'' \cite{NU} (hence the subscript ``I'').

Another measure of nonequilibrium is the quantity \footnote{We explain in Supplementary Note~3 how to assign a meaningful value to this quantity when $A$ is singular.}
\be\nu_\mathrm C(\rho):=1+\vc x^\dagger A^{-1}\vc x-\alpha.\ee
This quantity is also zero when $\rho=\gamma_\cS$, and non-zero for other states. However, it relates with the coherences present in the state (hence the subscript ``C''). The following result formalizes these quantities as measures of nonequilibrium.
\begin{theorem}\label{thgp}
$\nu_\mathrm I$ and $\nu_\mathrm C$ are non-increasing under low-temperature Gibbs-preserving operations.
\end{theorem}
These quantities, which are among a more general family described in Refs.~\cite{MJR14,CST}, are examples of monotones under the allowed operations. They can be identified by characterizing the Kraus operator representations of Gibbs-preserving operations (details in Supplementary Section~\ref{gmon}). In fact, since all cooling maps are Gibbs-preserving, these quantities are monotones also under cooling maps and low-temperature thermal operations.

These monotones together constitute sufficient conditions for state transitions under low-temperature Gibbs-preserving operations in the case where $\cS$ is a two-level system, i.e., $d=2$ (Supplementary Proposition~\ref{twolg}). They also turn out to be sufficient when both $\rho$ and $\sigma$ are pure (Supplementary Corollary~\ref{gpure}).

\subsection*{Comparing different thermodynamical models}
In particular, the two-level case provides a platform (see Fig.~\ref{figc}) to compare the Gibbs-preserving model with the exact treatment of thermal operations (which are equivalent to cooling maps for two-level systems). A host of state transitions that are forbidden under thermal operations are nonetheless allowed under Gibbs-preserving operations. This implies that the monotones $\nu$, when applied to thermal operations, are strictly less informative than the conditions of Theorem~\ref{thth}. The gaping disparity between the two models, which was first demonstrated in the recent work of Faist \textit{et al.}~\cite{Gib}, brings to the fore an important question in this field: which of the two models is a more accurate description of reality? While that dilemma remains, we now know the exact state transition conditions for cooling maps, which are demonstrably closer to low-temperature thermal operations than Gibbs-preserving operations are. Therefore, if one were to consider thermal operations the best available thermodynamical model, and if one considered the Gibbs-preserving model as an approximation thereto, then our work shows that we could have a better shot at finding exact state transition conditions by exploring classes of processes (such as our cooling maps) that are better approximations than the Gibbs-preserving model.
\begin{figure}[t]
    \centering
    \includegraphics[width=\columnwidth]{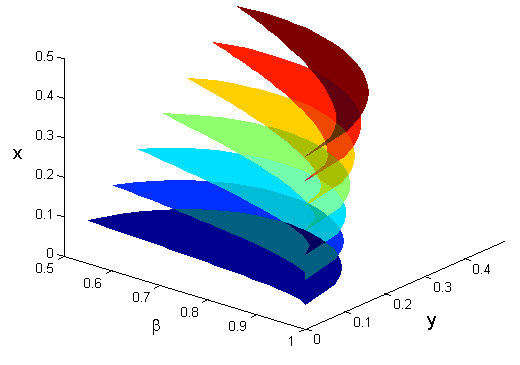}
    \caption{[Gibbs-preserving operations $\supsetneq$ cooling maps] Consider a parametric family of initial states $\rho(x):=1/2\proj{E_1}+x(\ket{E_1}\bra{E_2}+\ket{E_2}\bra{E_1})+1/2\proj{E_2}$, and a two-parameter family of final states $\sigma(y,\beta):=\beta\proj{E_1}+y(\ket{E_1}\bra{E_2}+\ket{E_2}\bra{E_1})+(1-\beta)\proj{E_2}$, on a two-level system, with $x,y,\beta$ real and nonnegative. For each value of $x$, the corresponding region in the $(y,\beta)$ plane represents part of the parametric state space that is reachable via Gibbs-preserving operations, but not via cooling maps (or thermal operations), from the initial state $\rho(x)$.}
    \label{figc}
\end{figure}

\section{Discussion}
Much remains to be discovered in the world of quantum thermodynamics. In particular, low-temperature situations, wherein exotic coherent phenomena lead to numerous technological applications, call for a thorough understanding of quantum coherences in thermodynamic processes. Some existing works on this aspect \cite{catcoh,CorW} pertain to the use of environmental coherence to aid thermodynamic state transitions in the system, as opposed to the evolution of the system's own coherence under state transitions. Recent work on the latter \cite{CohRE,MJR14,CorW,Coh,CST} provides insights that apply to all temperatures. Particularly, Ref.~\cite{CST} identifies the essence of this problem, namely the interplay between time-translation symmetry and thermal inequilibrium. However, in the general case it appears to be challenging to find complete (necessary and sufficient) conditions that can be expressed succinctly. In this paper, we compromise on the range of temperatures in our scope of generality, but by doing so we make significant progress in the low-temperature regime through our ``cooling maps'' characterization. We find the necessary and sufficient conditions for state transitions under cooling maps, and also confirm rigorously that low-temperature thermal operations can optimally preserve coherences.

The main open question emerging from this work is whether the mathematically characterized cooling maps are equivalent to the physically motivated thermal operations, or merely a close approximation thereof. Their equivalence for the cases of two-level systems and mixtures of optimally coherent processes motivates us to conjecture equivalence in general. The study of cooling maps aided by catalysts, and possible generalizations to higher temperatures, are other open problems that would provide insight into thermodynamics. Likewise, the monotones derived from the Gibbs-preserving model could have higher-temperature generalizations that improve our understanding of coherence transfer in thermodynamics. Finally, there is potential for experimental realization and testing of our results, for instance using a superconducting artificial atom coupled with a network of spins that acts as a bath. We leave these avenues for future work.

\makeatletter
\newcommand\tection{\@startsection {section}{1}{\z@}%
                                  {2ex}%
                                   {1ex}%
                                   {\noindent\normalsize\bfseries}}% from \Large
\makeatother

\tection*{Acknowledgments}
\noindent The authors thank Ish Dhand, Joydip Ghosh, Mark Girard, Micha{\l} Horodecki, Peter H\o yer, Barry Sanders, and Dongsheng Wang for helpful discussions. VN thanks Poornima Ambalavanan for help with figures. This work was carried out with the help of NSERC funding.
\widetext
\clearpage
\begin{center}
\textbf{\large Supplemental Material}
\end{center}
%%%%%%%%%% Merge with supplemental materials %%%%%%%%%%
%%%%%%%%%% Prefix a "S" to all equations, figures, tables and reset the counter %%%%%%%%%%
%\numberwithin{equation}{section}
\renewcommand{\thepage}{S\arabic{page}}
\renewcommand\thesection{S\arabic{section}}
\setcounter{page}{1}
\setcounter{section}{0}
\makeatletter
\renewcommand{\theequation}{S\arabic{section}.\arabic{equation}}
\addtocounter{ifsup}{1}
\numberwithin{figure}{ifsup}
\renewcommand{\thefigure}{S\arabic{figure}}
\setcounter{figure}{0}
\renewcommand{\bibnumfmt}[1]{[S#1]}
\renewcommand{\citenumfont}[1]{S#1}

\section{The thermal operations model}\label{sther}
Here we provide a summary of the relevant background for understanding the ``thermal operations" model of quantum thermodynamics. We base the discussion on the content of Refs.~\cite{SReth,SNan}.

Let us call the system of interest $\cS$. In classical thermodynamics, $\cS$ is some composite system consisting of a huge number of constituent parts---a gas, a spin lattice, etc. In that case we can accurately model thermal properties using a formalism that does not actually monitor the exact quantum state (the ``microstate") of $\cS$, but rather only a coarse-grained description that includes only a few so-called ``macroscopic" variables, such as the temperature, pressure, and magnetic moment. On the other hand, in quantum thermodynamics, the microstate is part of the formalism. The ``thermodynamic" element lies in how the \emph{environment} is modeled: The environment is assumed to be an ideal thermal reservoir (or ``heat bath"). This form of the environment, characterized by some properties that we will discuss below, naturally renders the dynamics of the system ``thermalizing".

This approach allows us to not only match the classical thermodynamical expectation of eventual ``equilibration" of the system with the environment, but to also understand how the microstate evolves \emph{while the system equilibrates}. The processes that can occur in the course of equilibration are classified under the label ``thermal operations".
\subsection{The heat bath}\label{bath}
The environment (call it $\cR$) of $\cS$ is an ideal heat bath, characterized by the following properties:
\begin{enumerate}
\item The state of $\cR$ is a Gibbs state at some temperature $T$. This temperature acts as the ``ambient" condition determining the dynamics of $\cS$.
\item This state of $\cR$ is supported almost entirely on a typical set $\cE_\cR$ of energy levels.
\item The energies $F$ in the typical set $\cE_\cR$ are concentrated in a region of radius $O\left(F_M^{1/2}\right)$ around the mean value $F_M$. (We use $F$ for energy levels of $\cR$, to distinguish them from those of $\cS$).
\item The multiplicity, or degeneracy, $g_\cR(F)$ of energy levels in $\cE_\cR$ scales at least exponentially in $F$:
\be g_\cR(F)\ge g_1\exp\left[c_\cR(F-F_1)\right]\ee
for some constant $c_\cR>0$, where $F_1$ is the ground state energy.
\item For any two energies $(E_i,E_j)$ of $\cS$, there exist $(F_k,F_\ell,)$ in $\cE_\cR$ such that\label{med}
\be E_i-E_j=F_k-F_\ell.\ee
\item For small perturbations about the peak $F_M$, the multiplicity goes as
\be g_\cR(F_M-\epsilon)\approx g_\cR(F_M)\exp(-\beta\epsilon),\ee
where $\beta:=1/(k_BT)$ with $k_B$ the Boltzmann constant.
\end{enumerate}
All of these properties are exhibited by a system that consists of many weakly interacting identical systems all prepared in their respective Gibbs states, i.e., a composite in a state of the form $\gamma^{\otimes n}$ with $\gamma$ a Gibbs state. In the present work, we are interested in the low-temperature limit. In this limit, the state of the bath is almost completely supported in its ground space, and therefore, all the above requirements excepting No.~\ref{med} are trivially satisfied.
\subsection{Carrying out a thermal operation}
We now consider the definition of thermal operations in detail, in order to clarify and justify the specifics. For convenience, we repeat below the definition of thermal operations from the text, with minor modifications.
\begin{definition}[Thermal operation]
A process (i.e., a quantum channel) on $\cS$, that can be realized operationally in the following steps:
\begin{enumerate}
\item Bring $\cS$ (which is initially isolated) together with an arbitrary ancillary system $\cA$, which is prepared in its own Gibbs state $\gamma_\cA:=\left(1/Z_\cA\right)\exp\left(-\beta H_\cA\right)$ corresponding to its own free Hamiltonian $H_\cA$ and the ambient temperature $T$. Physically, $\cA$ is all or part of the heat bath $\cR$, which in turn is modeled as discussed in the previous section.
\item Perform any global energy-preserving unitary evolution $U$ on the composite system $\cS\cA$.
\item Discard the ancilla $\cA$ (i.e., isolate $\cS$ again).
\end{enumerate}
Mathematically, the channel is represented by a completely positive (CP) trace-preserving (TP) map $\cE$ whose action on an arbitrary state $\rho$ of $\cS$ is given by
\begin{equation}\label{thop1}
\rho\mapsto\cE(\rho)=\Tr_\cA\left[U\left(\rho\otimes\gamma_\cA\right)U^\dagger\right],
\end{equation}
where $\Tr_\cA$ is the mathematical operation of partial trace with respect to $\cA$, corresponding to the physical operation of discarding the system $\cA$.
\end{definition}

Let us look closely at the above operational description: What does it mean to be able to attach an \emph{arbitrary} ancilla and perform an \emph{arbitrary} energy-conserving unitary? The arbitrariness of the ancilla $\cA$ means that the ancilla can feature \emph{any} number of degrees of freedom, and that its free Hamiltonian $H_\cA$ is unrestricted. $H_\cA$ could even be time-dependent: as explained in Ref.~\cite{SReth}, we can model time-dependence by a time-independent Hamiltonian, provided we include an additional ``clock" system into the apparatus. But what about interactions between $\cS$ and $\cA$? The fact that we start out and end up with $\cS$ isolated implies that, while we can ``turn on" an interaction in between, the initial and final settings must be ones where the dynamics of $\cS$ is free. Therefore, the Hamiltonian of the composite $\cS\cA$ at the start and end of the protocol has the form
\be H_{\cS\cA}=H_\cS\otimes\eins_\cA+\eins_\cS\otimes H_\cA.\ee
As explained in the main text, the energy conservation condition on the unitary evolution $U$ can be stated in terms of the eigenvalues and eigenvectors of $H_{\cS\cA}$ as
\be\bra{G_j;\alpha}U\ket{G_k;\beta}=0,\ee
where $G_j$ and $G_k$ are distinct eigenvalues. Also recall from the main text that the energy levels of $H_{\cS\cA}$ have the form
\be G_i=E_j+F_k,\ee
where $E_j$ is one of the eigenvalues of $H_\cS$ and $F_k$ an eigenvalue of $H_\cA$. An energy-conserving unitary can connect different energy levels on $\cS$ by raising or lowering $E$ while lowering or raising $F$ by the same amount.

\subsection{The low-temperature assumption}\label{slowt}
Here we make our notion of lowness of temperature more precise. We define low temperature with reference to the properties of the heat bath $\cR$, discussed earlier. One of the properties is that the state of the bath is a Gibbs state at some temperature $T$. This has the form
\begin{align}
\gamma_\cR&=\fr1{Z_\cR}\exp\left(-\beta H_\cR\right)\nonumber\\
&=\sum_j\fr{\exp\left(-\beta F_j\right)}{Z_\cR}\sum_{t=1}^{g_j}\proj{F_j;t}\nonumber\\
&=\sum_j\fr{g_j\exp\left(-\beta F_j\right)}{Z_\cR}\Pi_j.
\end{align}
Here we denote by $g_j$ the multiplicity of level $F_j$, and $t$ is some label that identifies individual eigenvectors within a degenerate subspace. $\Pi_j:=(1/g_j)\sum_{t=1}^{g_j}\proj{F_j;t}$ represents the normalized projector onto the subspace of energy $F_j$. If we now choose $\beta$ large enough that
\be g_1\exp\left(-\beta F_1\right)\gg g_j\exp\left(-\beta F_j\right)\ee
for any $j\ne1$, we then effectively have
\be\gamma_\cR\approx\Pi_1,\ee
which is the form in which the low-temperature assumption is used in the main matter. The range of temperatures at which this approximation is justified is determined by the nature of the bath, and also by the relation of the bath to the system. For example, for a bath consisting of many identical systems in identical Gibbs states, i.e. of the form $\gamma^{\otimes n}$, our low-temperature assumption is satisfied for temperatures $T\ll F_2-F_1$, where $F_1$ and $F_2$ are the ground and first excited state energies of each subsystem in the bath. Interestingly, it might be possible to justify our low-temperature assumption even in cases where we know little about the actual composition of the bath: based on the behavior of the system itself. For example, if the system is a superconducting circuit and the bath is the environment that is not in our control, then at temperatures below the system's superconducting critical temperature one could assume the bath to be in its ground state. This is because the system's existence in the superconducting phase implies that no energy is flowing from the bath into the system. In the remainder, we will use the term ``thermal operation" to mean ``thermal operation under the low-temperature assumption".

\section{Characterizing thermal operations as ``cooling maps"}\label{scool}
As we discussed in the previous supplementary note, our low-temperature assumption leads to the property that the initial state of any ancillary system $\cA$ used in implementing a thermal operation is supported almost entirely on its lowest energy level $F_1$:
\begin{equation}\label{lowt3}
\gamma_\cA\approx\left(\fr1{g_1}\right)\sum_{t=1}^{g_1}\proj{F_1;t}.
\end{equation}
In this note we will see that this leads to a convenient mathematical model.

\subsection{Cooling maps: motivation}\label{com}
Let us now turn our attention to the system of interest, $\cS$. It is characterized by its Hamiltonian $H_\cS$. Recall from the main text the following assumptions about $H_\cS$:
\begin{enumerate}
\item $H_\cS$ has no degenerate energy levels. Thus, its energy spectrum has the structure
\be E_1<E_2<\dots<E_d,\ee
where $d$ is the number of degrees of freedom in $\cS$.
\item For any two pairs of indices, $(i,j)$ and $(k,l)$,
\be E_i-E_j\ne E_k-E_l,\ee
except when either $i=j$ and $k=l$, or $i=k$ and $j=l$.\label{ndeg}
\end{enumerate}
These assumptions may seem very artificial and restrictive, but are in fact satisfied by generic physical systems. If a Hermitian matrix were chosen at random and assigned to act as the Hamiltonian, then with probability $1$ it would have the above properties. One might argue that actual physical systems, such as atoms, don't occur with random Hamiltonians, and typically have degenerate levels and gaps. But these degeneracies exist only when the systems are perfectly isolated from all external influences (e.g. electromagnetic fields). In reality the degeneracies are broken, even if only by tiny perturbations. Furthermore, the gaps  nature of such degeneracy-breaking phenomena, such as the Stark effect and the Zeeman effect, 

Recalling Eq.~(\ref{thop1}), and using the approximation Eq.~(\ref{lowt3}), we can write any thermal operation as
\begin{align}\label{thopd}
\cE(\rho)&=\Tr_\cA\left[U\left(\rho\otimes\gamma_\cA\right)U^\dagger\right]\nonumber\\
&\approx\Tr_\cA\left[U\left(\rho\otimes\left[\fr1{g_1}\sum_{t=1}^{g_1}\proj{F_1;t}\right]\right)U^\dagger\right]\nonumber\\
&=\fr1{g_1}\sum_{t=1}^{g_1}\Tr_\cA\left[U\left(\rho\otimes\proj{F_1;t}\right)U^\dagger\right]\nonumber\\
&=\fr1{g_1}\sum_{t=1}^{g_1}\cE_t(\rho),
\end{align}
where each $\cE_t$ is a CPTP map defined through
\be\cE_t(\rho):=\Tr_\cA\left[U\left(\rho\otimes\proj{F_1;t}\right)U^\dagger\right].\ee
The action of $\cE_t$ is determined by the action of $U$ on states of the form $\ket{E_j}\otimes\ket{F_1;t}$. In such a state, the energy of $\cS$ is $E_j$ while that of $\cA$ is the lowest possible, $F_1$. An energy-conserving $U$ can either retain the same amount of energy in either subsystem, or transfer some energy from $\cS$ to $\cA$. Therefore, level $j$ of $\cS$ can be mapped only to levels $k\le j$, and the overall effect is to ``cool" $\cS$.

It is useful to characterize the $\cE_t$'s through the structure of their Kraus operator decompositions. One possible set of Kraus operators $\{K_i\}$ can be constructed by assigning the following values to its matrix elements:
\be\bra{E_j}K_i\ket{E_k}:=\left(\bra{E_j}\otimes\bra{v_i}\right)U\left(\ket{E_k}\otimes\ket{F_1;t}\right),\ee
where
\be\{\ket{v_i}\}=\{\ket{F_\ell;s}\}.\ee
Physically, the above construction represents the fact that $K_i$ can change the state of $\cS$ from $\ket{E_k}\mapsto\ket{E_j}$ by virtue of $U$ taking the composite $\cS\cA$ from $\ket{E_k}\otimes\ket{F_1;t}\mapsto\ket{E_j}\otimes\ket{v_i}$. The $K_i$'s thus constructed fall into two categories:
\begin{enumerate}
\item When $\ket{v_i}=\ket{F_1;s}$ for some $s$: This case corresponds to $U$ not causing any flow of energy from $\cS$ to $\cA$ (since $\cA$ stays within the same energy level where it started). Because $H_\cS$ has no degeneracies, the final state of $\cS$, $\ket{E_j}$, must be identical with its initial state, $\ket{E_k}$. Therefore the $K_i$'s in this category are \emph{diagonal}.
\item When $\ket{v_i}=\ket{F_\ell;s}$ is an excited state of $\cA$: Here $U$ is raising $\cA$ from $F_1$ to $F_\ell\ne F_1$. Therefore, for energy conservation,
\be E_k-E_j=F_\ell-F_1.\ee
By the property \ref{ndeg} of $H_\cS$, there must be a unique pair $(j,k)$ satisfying this condition for a given $\ell$. Therefore, only one matrix element of such a $K_i$ can be non-zero, and so we arrive at the form
\be K_i\propto\ket{E_j}\bra{E_k}.\ee
\end{enumerate}
In the second category, note that $j$ is always smaller than $k$. Since each of the $\cE_t$'s can be Kraus-decomposed in this way, and $\cE$ is an incoherent mixture of the $\cE_t$'s, such a Kraus decomposition also exists for $\cE$. This suggests that probing the set of all channels with such Kraus decompositions might shed light on thermal operations. To this end, we define
\begin{definition}[Cooling map]
A quantum channel (CPTP map) with a Kraus decomposition consisting of Kraus operators of the following two classes:
\begin{enumerate}
\item Diagonal matrices $\{K_1\dots K_n\}$. Without loss of generality, we can assume $n\le d$.
\item Matrices of the form $J_{jk}\propto\ket{E_j}\bra{E_k}$, $j<k$. Without loss of generality we can assume that there is only one $J$ for every index pair $(j,k)$. For, if $\mu_{jk}\ket{E_j}\bra{E_k}$ and $\nu_{jk}\ket{E_j}\bra{E_k}$ are two Kraus operators occurring in the same decomposition of some channel, then we can combine them into just the one operator $\sqrt{\abs{\mu_{jk}}^2+\abs{\mu_{jk}}^2}\ket{E_j}\bra{E_k}$.
\end{enumerate}
\end{definition}
All matrix representations are in the standard basis $\{\ket{E_j}\}$. Note that the elements of the matrices can be complex. By the discussion preceding the above definition, we have the following:
\begin{obs}\label{obs1}
All low-temperature thermal operations are cooling maps.
\end{obs}

\subsection{The action of cooling maps}
Let us examine the action of a generic cooling map $\cE$ on a generic initial state $\rho$. Let a possible set of Kraus operators for $\cE$ be
\begin{align}
K_i&=\left(\begin{array}{cccc}\lambda_1^{(i)}&0&\hdots&0\\
0&\lambda_2^{(i)}&0&\vdots\\
\vdots&0&\ddots&0\\
0&\hdots&0&\lambda_d^{(i)}
\end{array}\right),\;i\in\{1\dots n\};\nonumber\\
J_{jk}&=\mu_{jk}\ket j\bra k,\;j<k\in\{1\dots d\}.\nonumber
\end{align}
Denote by $\vc\lambda_j$ the $n$-dimensional complex vector whose components are $\lambda_j^{(i)}$. Let $q$ be the \emph{Gramian matrix} of the collection $(\vc\lambda_1\dots\vc\lambda_d)$ of vectors. The Gramian is defined through
\be q_{jk}=\left\langle\vc \lambda_j,\vc \lambda_k\right\rangle,\ee
where on the right-hand side is the usual inner product between two vectors on $\bbC^n$. Define also the matrix $P\equiv(P_{j|k})$, through
\begin{equation}\label{Pdef}
P_{j|k}=\left\{\begin{array}{ll}q_{jj},&\textnormal{if }j=k;\\
\abs{\mu_{jk}}^2,&\textnormal{if }j\ne k.\end{array}\right.
\end{equation}
It can be seen by inspection that the action of $\cE$ on $\rho$ yields the state $\sigma$ whose components are given by
\begin{equation}\label{acte}
\sigma_{jk}=\left\{\begin{array}{ll}\sum_{\ell=1}^dP_{j|\ell}\rho_{\ell\ell},&\textnormal{if }j=k;\\
q_{jk}\rho_{jk},&\textnormal{if }j\ne k.\end{array}\right.
\end{equation}

The matrix $P$ has the following properties:
\begin{enumerate}
\item Upper-triangularity: $P_{j|k}=0$ if $j>k$. This follows from the upper-triangularity of the $J$'s.
\item Column-stochasticity: $P_{j|k}\ge0$  for all $(j,k)$; and $\sum_{j=1}^dP_{j|k}=1$ for all $k$. The latter follows from the trace-preserving (TP) condition on the action of $\cE$ [Eq.~(\ref{acte})]. The stochasticity of $P$ is the motivation for our use of ``conditional probability" notation to denote its matrix elements.
\end{enumerate}

In connection with the Gramian of a set of vectors, we recall the following useful result from linear algebra \cite{SHJ}: For any collection $(\vc v_j)$ of vectors on an inner product space, the Gramian matrix $q$ of the collection is positive-semidefinite. Conversely, any positive-semidefinite matrix is the Gramian of some collection of vectors. Combining this fact with the preceding observations about the action of cooling maps leads to:
\begin{lemma}\label{eqv}
For any two states $(\rho,\sigma)$ of $\cS$, the existence of a cooling map $\cE$ mapping $\rho\mapsto\sigma$ is equivalent to the existence of a $d\times d$ positive-semidefinite matrix $q$ with the following properties:
\begin{enumerate}
\item The diagonal of $q$ must be identical with the diagonal of an upper-triangular column-stochastic matrix $P$ such that
\be(\sigma_{11}\dots\sigma_{dd})^T=P(\rho_{11}\dots\rho_{dd})^T.\ee
\item Each off-diagonal element $q_{jk}$ must satisfy
\be\sigma_{jk}=q_{jk}\rho_{jk}.\ee
\end{enumerate}
\end{lemma}

\subsection{Upper-triangular stochastic matrices and majorization}
It will be useful for our present purpose to better understand upper-triangular column-stochastic (UTCS) matrices. General column-stochastic matrices are known to induce a \emph{preorder} on the set of probability distributions, called the \emph{majorization} preorder \cite{SMOA}. In the following lemma, we prove that the action of UTCS matrices induces a \emph{partial order}, which by analogy we name ``upper-triangular majorization", or ``UT majorization".
\begin{definition}[UT majorization]
Let $\vc u\equiv(u_1,u_2\dots u_d)^T$ and $\vc v\equiv(v_1,v_2\dots v_d)^T$ be two $d$-dimensional probability distributions. We say that $\vc u$ \emph{UT-majorizes} $\vc v$, denoted $\vc u\dsm\vc v$, if the following $(d-1)$ inequalities are satisfied:
\begin{align}
u_d&\ge v_d,\nonumber\\
u_{d-1}+u_d&\ge v_{d-1}+v_d,\nonumber\\
&\;\;\vdots\nonumber\\
u_2+u_3\dots+u_d&\ge v_2+v_3\dots+v_d.\nonumber
\end{align}
\end{definition}
\begin{lemma}\label{Pcon}
If $\vc u$ and $\vc v$ are $d$-dimensional probability vectors and there exists a UTCS matrix $P$ such that $\vc v=P\vc u$, then $\vc u\dsm\vc v$.

Conversely, if $\vc u\dsm\vc v$, then there exists a UTCS $P$ such that $\vc v=P\vc u$. In fact, there exists such a $P$ with the following specific values on its diagonal:
\be P_{j|j}=\left\{\begin{array}{ll}\min\left(\fr{v_j}{u_j},1\right),&\textnormal{if }u_j>0;\\
0,&\textnormal{if }u_j=0.\end{array}\right.\ee
\end{lemma}
\begin{proof}
Assume first that there exists a UTCS $P$ such that $\vc v=P\vc u$. Componentwise, we have
\begin{align}
v_d=&P_{d|d}u_d;\nonumber\\
v_{d-1}=P_{d-1|d}u_d&+P_{d-1|d-1}u_{d-1};\nonumber\\
&\;\;\vdots\nonumber\\
v_1=P_{1|d}u_d+P_{1|d-1}&u_{d-1}\dots+P_{1|1}u_1.\nonumber
\end{align}
The stochasticity of $P$ implies that each of its elements is no greater than $1$ (i.e., $P_{j|k}\le1$). Therefore, the first of the above equations implies that $v_d\le u_d$. Adding the first two equations, we get $v_{d-1}+v_d\le u_{d-1}+u_d$. Continuing in this manner, we have all the desired inequalities to prove $\vc u\dsm\vc v$.\qed

Now to prove the converse, assume that $\vc u\dsm\vc v$. We shall construct a $P$ with the desired properties. Firstly, we fix the diagonal elements of $P$ as claimed in the Lemma statement:
\be P_{j|j}=\left\{\begin{array}{ll}\min\left(\fr{v_j}{u_j},1\right),&\textnormal{if }u_j>0;\\
0,&\textnormal{if }u_j=0.\end{array}\right.\ee
By construction, these values lie in the interval $[0,1]$ and so we're on track to construct a stochastic $P$. For each $j$, we require $P$ to act in such a way that
\begin{equation}\label{vPu}
v_j=P_{j|d}u_d+P_{j|d-1}u_{d-1}\dots+P_{j|j}u_j.
\end{equation}
The last term of the RHS, $P_{j|j}u_j$, is already fixed by our definition of the diagonal element $P_{j|j}$. It remains to choose the $P_{j|k}$ for all $k>j$ in such a way as to satisfy the above equation. The freedom we have in this choice is characterized by the quantity
\be r_j=v_j-P_{j|j}u_j=\max\left(0,v_j-u_j\right),\ee
which we may think of as a ``remainder" or ``deficit": the part of the RHS of Eq.~(\ref{vPu}) that remains to be filled in. Now let us consider each $j$ in sequence, starting from $j=d$.

The premise $\vc u\dsm\vc v$ implies that
\be u_d\ge v_d,\ee
and therefore,
\be r_d=0.\ee
This means that Eq.~(\ref{vPu}) has been achieved for $j=d$. The part of $u_d$ that is still ``available" to be mapped to lower components of $\vc v$ is
\be a_d:=u_d\left(1-P_{d|d}\right)=u_d-v_d\ge0.\ee
Now consider $j=d-1$. Again, $\vc u\dsm\vc v$ implies
\begin{align}
u_{d-1}+u_d&\ge v_{d-1}+v_d.\nonumber\\
\Rightarrow v_{d-1}-u_{d-1}&\le a_d.\nonumber
\end{align}
But the deficit in the $(d-1)^\textnormal{th}$ component is
\be r_{d-1}=\max\left(0,v_{d-1}-u_{d-1}\right)\le a_d.\ee
Therefore, this deficit can be filled in by some part of $a_d$. We do this by assigning
\be P_{d-1|d}:=\fr{r_{d-1}}{a_d}\left(1-P_{d|d}\right).\ee
The components of $P$ assigned thus far have taken care of Eq.~(\ref{vPu}) for $j=d$ and $j=d-1$. The part of $(u_{d-1}+u_d)$ that is still available to be mapped to lower components of $\vc v$ is
\begin{align}
a_{d-1}&:=u_d\left(1-P_{d|d}-P_{d-1|d}\right)+u_{d-1}\left(1-P_{d-1|d-1}\right)\nonumber\\
&=u_d+u_{d-1}-v_d-v_{d-1},\nonumber
\end{align}
and again, $\vc u\dsm\vc v$ implies that $a_{d-1}\ge0$.
\begin{figure}
\centering
\includegraphics[width=\textwidth]{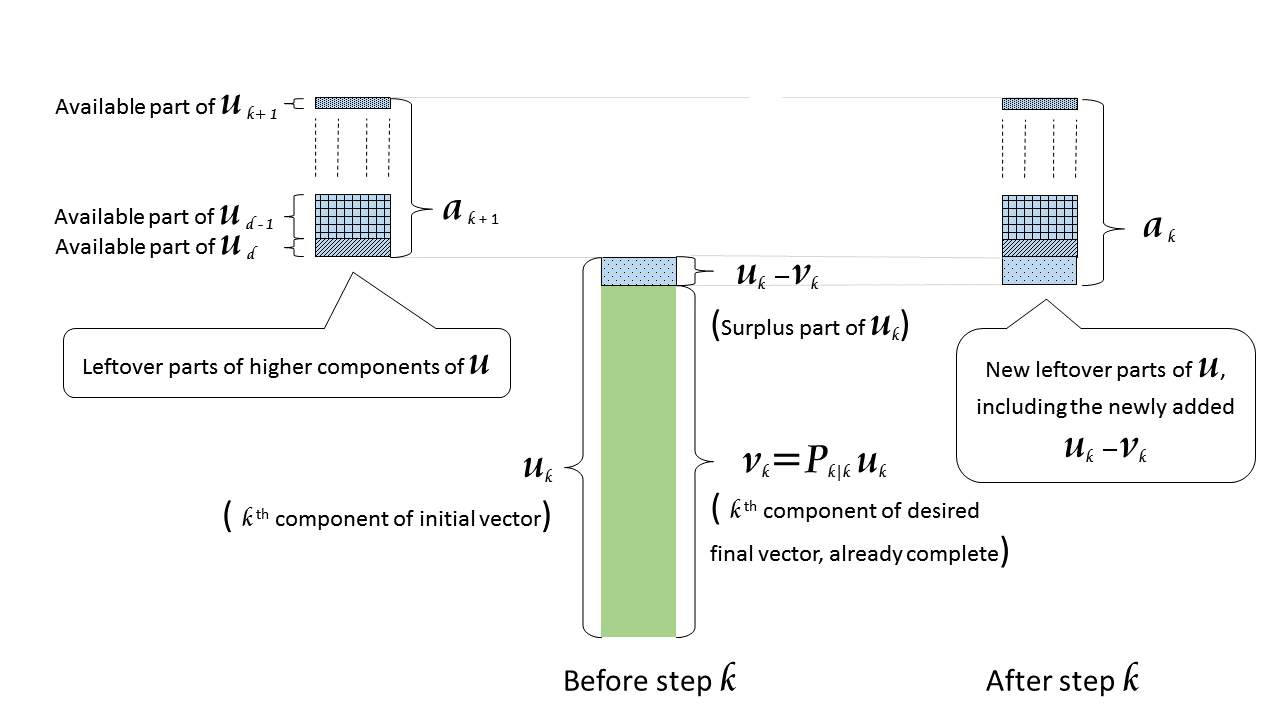}
\caption{A pictorial depiction of the construction of the stochastic map mapping $\vc u\mapsto\vc v$ in the proof of Lemma~\protect\ref{Pcon}. If $u_k\ge v_k$, then the $k^\textnormal{th}$ (backwards!) step consists of mapping the fraction $P_{k|k}=v_k/u_k$ of $u_k$ to complete the desired $v_k$ and then adding the leftover part $(u_k-v_k)$ to $a_{k+1}$ to enhance it to $a_k$.}\label{proof1}
\end{figure}
\begin{figure}
\centering
\includegraphics[width=\textwidth]{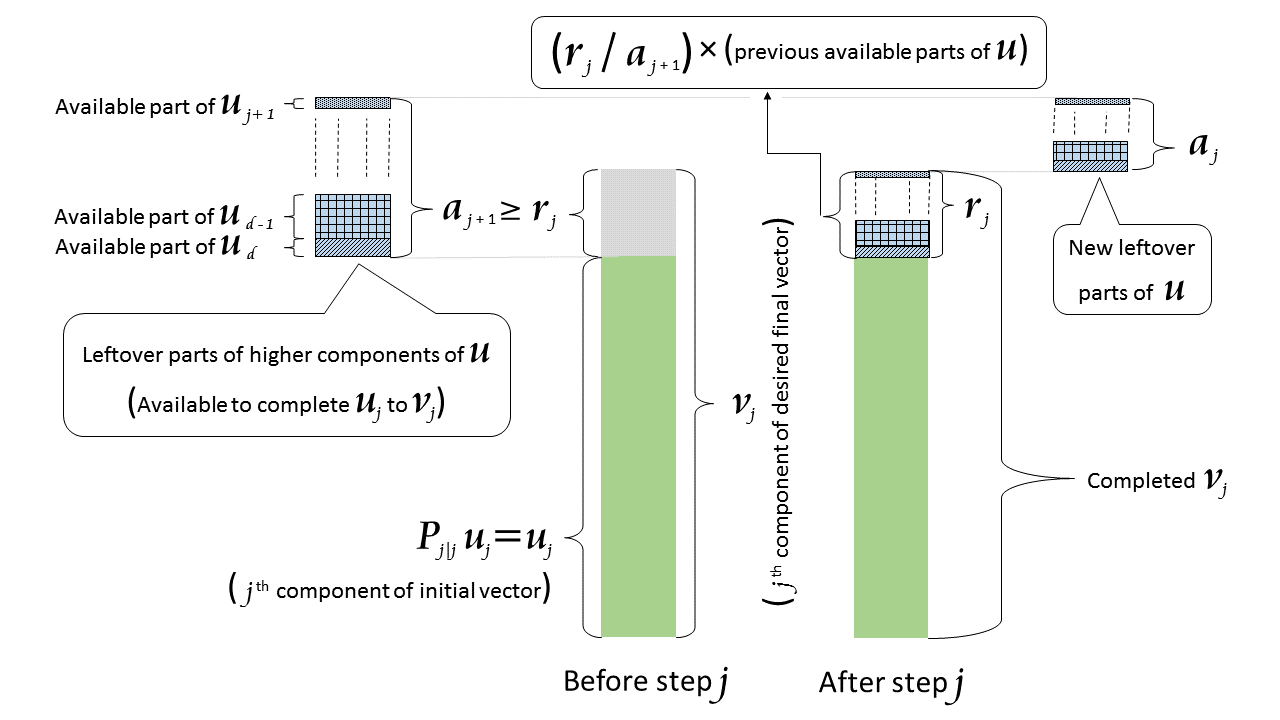}
\caption{In the $j^\textnormal{th}$ step, if $u_j<v_j$, then $P_{j|j}=1$, and $P_{j|j}u_j=u_j$ is still smaller than the desired $v_j$ by $r_j=v_j-u_j$. We then add exactly the fraction $(r_j/a_{j+1})$ of the leftover part of each higher component of $\vc u$ to $u_j$ to complete it to $v_j$.}\label{proof2}
\end{figure}

In the next step we have again that
\be r_{d-2}\le a_{d-1}\ee
and can therefore carry out a similar procedure as before, assigning
\be P_{d-2|d}:=\fr{r_{d-2}}{a_{d-1}}\left(1-P_{d|d}-P_{d-1|d}\right)\ee
and
\be P_{d-2|d-1}:=\fr{r_{d-2}}{a_{d-1}}\left(1-P_{d-1|d-1}\right).\ee
The basic idea is the following: for any $k$, if $u_k\ge v_k$, then $r_k=0$ and $P_{k|k}u_k=v_k$, therefore we do not need to map any higher component ($u_\ell$ for $\ell>k$) to ``complete" $v_k$. We can in fact add the surplus part $u_k-v_k$ to the ``available" $a_{k+1}$ to get a larger number, $a_k$, that is now available to complete the $v_j$'s for $j<k$. Supplementary Figure~1 illustrates this idea.

On the other hand, if for some $j$, $u_j<v_j$ (Supplementary Figure~2), then $r_j>0$ and $P_{j|j}u_j<v_j$. But in such a case, thanks to the UT majorization condition, we are assured that the ``available" part left over from higher components, which by our convention we call $a_{j+1}$, is at least $r_j$. We then take the overall part left over from each higher component of $\vc u$ and use up exactly the fraction $r_j/a_{j+1}$ of it to complete the $j^\textnormal{th}$ instance of Eq.~(\ref{vPu}), i.e.,
\be v_j=\sum_{k\ge j}P_{j|k}u_k.\ee
$P$ is upper-triangular by construction. Furthermore, using the recursive definition of the components of $P$, we can verify that
\be P_{j|k}\ge0\ee
and
\be\sum_{k\le j}P_{k|j}=1,\ee
guaranteeing stochasticity.
\end{proof}

\subsection{Reality check: UT majorization emerges from thermo-majorization}\label{sreal}
In quantum thermodynamics at general temperatures, an ordering relation called \emph{thermo-majorization} \cite{SNan} plays the role corresponding to that of UT majorization in our formalism. Although we arrived at UT majorization through rigorously examining the energy conservation condition in low-temperature thermal operations, it is worth while to convince ourselves of the soundness of our low-temperature limit. Why this matter is not trivial will become clear when we consider the following definition of thermo-majorization:
\begin{definition}[Thermo-majorization]
For $d$-dimensional probability distributions $\vc u$ and $\vc v$, $\vc u$ \emph{thermo-majorizes} $\vc v$, denoted
\be\vc u\thm\vc v,\ee
if there exists a \emph{column-stochastic} matrix $P$ such that
\begin{enumerate}
\item $P$ fixes the Gibbs distribution: $P\vc u_\gamma=\vc u_\gamma$, where $\vc u_\gamma:=\left(1/Z_\cS\right)\left(\exp(-\beta E_1)\dots\exp(-\beta E_d)\right)^T$ is the diagonal part of the Gibbs state $\gamma_\cS$.
\item $P$ maps $\vc u$ to $\vc v$: $\vc v=P\vc u$.
\end{enumerate}
\end{definition}
Ostensibly, it might seem that the low-temperature limit of thermo-majorization could be obtained by simply approximating the Gibbs state by the ground state:
\be\gamma_\cS\approx\proj{E_1}.\ee
This approximation would lead to a corresponding counterpart of thermo-majorization that is associated with all stochastic matrices $P$ that obey
\be P_{j|1}=0\ee
for $j>1$. However, this is clearly different from UT majorization, which is associated with a more restricted class of such $P$'s---namely, upper-triangular matrices. The following exercise serves to vindicate UT majorization as the right option in favour of the less-restrictive version. Consider some finite inverse temperature $\beta$. We then have the following conditions for $P$ to fix $\vc u_\gamma$:
\begin{align}
P_{1|1}\exp(-\beta E_1)+\sum_{j>1}P_{1|j}\exp(-\beta E_j)&=\exp(-\beta E_1),\nonumber\\
P_{2|1}\exp(-\beta E_1)+P_{2|2}\exp(-\beta E_2)+\sum_{j>2}P_{2|j}\exp(-\beta E_j)&=\exp(-\beta E_2),\nonumber\\
&\;\;\vdots\nonumber\\
\sum_{j<d}P_{d|j}\exp(-\beta E_j)+P_{d|d}\exp(-\beta E_d)&=\exp(-\beta E_d).\label{pgam}
\end{align}
In the limit $\beta\to\infty$,
\be\exp\left[-\beta(E_j-E_k)\right]=0\ee
whenever $j>k$. In this limit if we multiply the $j^\textnormal{th}$ of Eqs.(~\ref{pgam}), for any $j>1$, by $\exp(\beta E_1)$, we end up with
\be P_{j|1}=0\;\forall j>1.\ee
Now considering only the equations for $j>2$, we multiply by $\exp(\beta E_2)$ to infer that
\be P_{j|2}=0\;\forall j>2.\ee
Proceeding in this manner, we can prove that $P$ is upper-triangular in the limit.

One can carry out a similar verification with the other, equivalent definition of thermo-majorization in Ref.~\cite{SNan} (in terms of Gibbs-rescaled and reordered distributions). There one will find that for all distributions with no zero entries (i.e., for all but a measure-zero subset) the canonical permutation of vector components through which thermo-majorization is defined will approach the identity permutation as $\beta\to\infty$, thereby yielding UT majorization in the limit.

\subsection{State transition conditions}\label{stran}
We now have all the ingredients to derive our main result: the necessary and sufficient conditions for a state transition to be achievable through a cooling map.
\begin{athm}[Theorem~\ref{thth} of the main text]\label{coth}
For two states $\rho$ and $\sigma$ on $\cS$, arbitrary except that the matrix elements of $\rho$ are all nonzero ($\rho_{jk}\equiv\bra{E_j}\rho\ket{E_k}\ne0$), define the $d\times d$ matrix $Q$:
\be Q_{jk}=\left\{\begin{array}{ll}\min\left(\fr{\sigma_{jj}}{\rho_{jj}},1\right),&\textnormal{if }j=k;\\
\fr{\sigma_{jk}}{\rho_{jk}},&\textnormal{if }j\ne k.\end{array}\right.\ee
The state transition $\rho\mapsto\sigma$ is possible through a cooling map \emph{if and only if} both of the following conditions hold:
\begin{enumerate}
\item The diagonal of $\rho$ UT-majorizes that of $\sigma$:
\be(\rho_{11}\dots\rho_{dd})^T\dsm(\sigma_{11}\dots\sigma_{dd})^T;\ee
\item The matrix $Q$ is positive-semidefinite:
\be Q\ge0.\ee
\end{enumerate}
\end{athm}
\begin{proof}[Proof$\Leftarrow$]
Assume that the conditions stated in the theorem hold. The second condition states that $Q\ge0$. From the first condition and Supplementary Lemma~\ref{Pcon}, it follows that the diagonal elements of $Q$ are the diagonal elements of a UTCS matrix that maps $(\rho_{11}\dots\rho_{dd})^T\mapsto(\sigma_{11}\dots\sigma_{dd})^T$. As well, the off-diagonal elements of $Q$ are constructed to satisfy the condition of Supplementary Lemma~\ref{eqv}. Therefore, by the same lemma, there exists a cooling map that takes $\rho$ to $\sigma$.
\end{proof}
\begin{proof}[Proof$\Rightarrow$]
Assume now that there exists a cooling map achieving $\rho\mapsto\sigma$. By Supplementary Lemma~\ref{eqv}, there exists a $d\times d$ matrix $q\ge0$ with the following properties:
\begin{enumerate}
\item The diagonal of $q$ is also the diagonal of a UTCS matrix $P$ such that \be(\sigma_{11}\dots\sigma_{dd})^T=P(\rho_{11}\dots\rho_{dd})^T;\ee
\item For $j\ne k$, $\sigma_{jk}=q_{jk}\rho_{jk}$.
\end{enumerate}
Then, we have the following arguments to prove the corresponding conditions stated in the theorem:
\begin{enumerate}
\item From the first condition above, it follows that there exists a UTCS $P$ that maps $(\rho_{11}\dots\rho_{dd})^T\mapsto(\sigma_{11}\dots\sigma_{dd})^T$. Therefore, by Supplementary Lemma~\ref{Pcon},
\be(\rho_{11}\dots\rho_{dd})^T\dsm(\sigma_{11}\dots\sigma_{dd})^T.\ee
\item Consider the matrix $Q$ defined in the theorem statement. It has the same off-diagonal elements as $q$, but the diagonal elements
\be Q_{jj}=\min\left(\fr{\sigma_{jj}}{\rho_{jj}},1\right).\ee
For any UTCS matrix $P$ that maps $(\rho_{11}\dots\rho_{dd})^T\mapsto(\sigma_{11}\dots\sigma_{dd})^T$, the diagonal elements are bounded as follows:
\be P_{j|j}\le\min\left(\fr{\sigma_{jj}}{\rho_{jj}},1\right).\ee
Therefore,
\be q_{jj}=P_{j|j}\le Q_{jj}.\ee
This implies that
\be Q=q+D,\ee
where $D$ is a diagonal matrix with nonnegative entries. Since $q$ and $D$ are both positive-semidefinite, it follows that
\be Q\ge0.\ee
\end{enumerate}
\end{proof}
We can adapt the above theorem to cases where one or more $\rho_{jk}$'s are zero. The following proposition contains the modified version.
\begin{prop}
In cases where there are one or more zeroes in the matrix representation of $\rho$, the conditions of the theorem are replaced by the following revised set of conditions. In addition to the revised version of the two original conditions there is a third one, which we list \emph{first} because it is the easiest to check (and not because we believe that any respectable theory of thermodynamics must have a ``zeroth" law):
\begin{enumerate}
\setcounter{enumi}{-1}
\item For each pair $(j,k)$ such that $j\ne k$ and $\rho_{jk}=0$, the corresponding entry in $\sigma$ is also zero, i.e. $\sigma_{jk}=0$.
\item The first of the original conditions of Theorem~\ref{coth} stays the same:
\be(\rho_{11}\dots\rho_{dd})^T\dsm(\sigma_{11}\dots\sigma_{dd})^T.\ee
\item Before we state the condition, note that the $Q$ as defined in the theorem has diverging terms. We first take the following steps to construct an alternate \emph{family} of $Q$'s:
\begin{enumerate}
\item For all pairs of indices $(j,k)$ for whom $\rho_{jk}\ne0$, use the original definition of $Q_{jk}$.
\item For every $j$ such that $\rho_{jj}=0$, assign the value $0$ to all $Q_{jk}$ and $Q_{kj}$ (i.e., to the entire $j^{\textnormal{th}}$ row and column).
\item For every pair $(j,k)$ such that $\rho_{jk}=0$ and $Q_{jk}$ has not been set to zero in the previous step, allow $Q_{jk}$ to take any value.
\end{enumerate}
The revised second condition is that \emph{at least one} set of assignments in the last step lead to $Q\ge0$. In this sense, instead of one specific $Q$, we would now have to check a range of different $Q$'s. To minimize the complexity of this check, without loss of generality we can restrict each $Q_{jk}$ in the last step to be real and within the interval $\left[-(Q_{jj}Q_{kk})^{1/2},(Q_{jj}Q_{kk})^{1/2}\right]$.
\end{enumerate}
\end{prop}

\subsection{Cooling maps and thermal operations}
Our motivation in constructing the cooling maps model was the fact that all (low-temperature) thermal operations are cooling maps (Observation~\ref{obs1}). Here we present some arguments that support the following conjecture:
\begin{conj}
Cooling maps are equivalent to low-temperature thermal operations, with regard to the feasibility of state transitions.
\end{conj}
Note that this could be true even if the set of cooling maps is strictly larger than that of thermal operations---there could still be a thermal operation achieving every state transition that is possible through cooling maps.

Consider some state transition $\rho\mapsto\sigma$ that is possible under cooling maps. By Theorem~\ref{coth} this corresponds to the existence of a certain $d\times d$ matrix $Q\ge0$ associated with a possible operator sum representation of a cooling map achieving the transition. Specifically, the diagonal Kraus operators in the representation are parametrized by a collection $(\vc\lambda_1\dots\vc\lambda_d)$ of vectors whose Gramian is $Q$. The $i^\textnormal{th}$ diagonal Kraus operator contains the $i^\textnormal{th}$ component of each of these vectors:
\be K_i=\left(\begin{array}{cccc}\lambda_1^{(i)}&0&\hdots&0\\
0&\lambda_2^{(i)}&0&\vdots\\
\vdots&0&\ddots&0\\
0&\hdots&0&\lambda_d^{(i)}
\end{array}\right).\ee
In addition, of course, there are the off-diagonal Kraus operators
\be J_{jk}=\mu_{jk}\ket j\bra k,\;j<k\in\{1\dots d\}.\ee
If the Gramian $Q$ has rank $g$, then a thermal operation implementation of $\cE$ must necessarily use an ancilla $\cA$ whose ground state has multiplicity \emph{at least $g$}. Recall Eq.~(\ref{thopd}): The action of a cooling map $\cE$ that uses an ancilla with a $g$-fold degenerate ground state can be written as a uniform mixture of $g$ CPTP maps in the following manner:
\be\cE(\rho)=\fr1{g}\sum_{t=1}^{g}\cE_t(\rho),\ee
where $\cE_t$ is defined as
\be\cE_t(\rho):=\Tr_\cA\left[U\left(\rho\otimes\proj{F_1;t}\right)U^\dagger\right].\ee
We can find a Kraus operator sum representation for each $\cE_t$ using the same principle as we did before:
\be\bra{E_j}K_{i(t)}\ket{E_k}:=\left(\bra{E_j}\otimes\bra{v_i}\right)U\left(\ket{E_k}\otimes\ket{F_1;t}\right),\ee
where $\{\ket{v_i}\}$ is an orthonormal basis on the space of the composite $\cS\cA$.

The task of finding a thermal operation implementation of $\cE$ boils down to the task of finding a single energy-conserving $U$ that can enable various $\cE_t$'s, which in turn are free to be any CPTP maps as long as their uniform mixture is the channel $\cE$. In some cases it is possible to construct a $U$ that makes each $\cE_t$ identical with $\cE$, thereby realizing the latter channel overall. In such a case, the same $Q$ is associated with all $\cE_t$'s, but the $\vc\lambda$'s themselves are not required to be fixed---we only require that their Gramian be $Q$. The Gramian of a collection of vectors is invariant under isometries, giving us some freedom to choose the Kraus operators that we use in decomposing $\cE$ for different $t$'s. Let $(\vc\lambda_{1(t)}\dots\vc\lambda_{d(t)})$ be the particular vectors that we use in the $t^\textnormal{th}$ decomposition. A $U$ that achieves this could plausibly (although not necessarily) act in the following manner:
\be U\left(\ket{E_k}\otimes\ket{F_1;t}\right)=\left(\sum_{s=1}^g\lambda_{k(t)}^{(s)}\ket{E_k}\otimes\ket{F_1;s}\right)+\left(\sum_{j<k}\mu_{jk}\ket{E_j}\otimes\ket{F_{jk};t}\right),\ee
where $F_{jk}-F_1=E_k-E_j$, and $\{\ket{F_{jk};1}\dots\ket{F_{jk};g}\}$ may be chosen to be an orthonormal set of eigenvectors in the energy level $F_{jk}$ (we are allowed to give arbitrary multiplicities to the energy levels of $H_\cA$, to suit our convenience).

The requirement that $U$ be unitary implies that the vectors $\{U\left(\ket{E_k}\otimes\ket{F_1;1}\right)\dots U\left(\ket{E_k}\otimes\ket{F_1;g}\right)\}$ be mutually orthogonal for each $k$. In terms of the $\vc\lambda$'s, this amounts to
\be\left\langle\vc\lambda_{k(t)},\vc\lambda_{k(s)}\right\rangle\propto\delta_{ts}.\ee
On the other hand, the Gramian of each collection $(\vc\lambda_{1(t)}\dots\vc\lambda_{d(t)})$ must be $Q$. This is equivalent to the requirement that these collections all be mutually connected by isometries. This condition can be phrased as a property of $Q$:
\begin{prty}\label{conjal}
For the given $d\times d$ matrix $Q$ of rank $g$, there exist $g$ sets of $d$ vectors each, indexed as $(\vc\lambda_{1(t)}\dots\vc\lambda_{d(t)})_{t\in\{1\dots g\}}$, such that
\be\left\langle\vc\lambda_{j(t)},\vc\lambda_{k(t)}\right\rangle=Q_{jk}\ee
for all $j,k\in\{1\dots d\}$ and $t\in\{1\dots g\}$, and
\be\left\langle\vc\lambda_{k(t)},\vc\lambda_{k(s)}\right\rangle\propto\delta_{ts}\ee
for all $k\in\{1\dots d\}$ and $s,t\in\{1\dots g\}$.
\end{prty}
For any $Q$ with this property, we can construct an energy-conserving $U$ as discussed above, therefore qualifying the associated cooling map as a (low-temperature) thermal operation.

It is easy to verify that Property~\ref{conjal} is possessed by any $Q$ in the case $d=2$. Thus we have the following.
\begin{coro}\label{c2l}
Cooling maps are equivalent to low-temperature thermal operations on two-level systems.
\end{coro}
Recently, \'Cwikli\'nski \textit{et al.} \cite{SCoh} found the conditions for two-level systems at any temperature. Our conditions match the low-temperature limit of theirs.

Another special case where Property~\ref{conjal} obviously follows is when $Q$ is diagonal, and correspondingly, the final state $\sigma$ in the associated thermal operation is diagonal. Therefore, the physical context of this special case is a process wherein the coherences present in the initial state are completely lost. Perhaps this is not a very useful sort of process, but the next special case lies at the opposite extreme, and is therefore---presumably---extremely useful.

If $Q$ has rank $1$, then again it is straightforward to see that Property~\ref{conjal} holds. In order to understand the physical significance of this special case, consider again a generic cooling map $\cE$ with Kraus operators
\begin{align}
K_i&=\left(\begin{array}{cccc}\lambda_1^{(i)}&0&\hdots&0\\
0&\lambda_2^{(i)}&0&\vdots\\
\vdots&0&\ddots&0\\
0&\hdots&0&\lambda_d^{(i)}
\end{array}\right),\;i\in\{1\dots n\};\nonumber\\
J_{jk}&=\mu_{jk}\ket j\bra k,\;j<k\in\{1\dots d\}.\nonumber
\end{align}
The effect of $\cE$ on the off-diagonal elements of states [cf. Eq.~(\ref{acte})] is given by
\be\rho_{jk}\mapsto\sigma_{jk}=\left\langle\vc\lambda_j,\vc\lambda_k\right\rangle\rho_{jk}.\ee
By the Cauchy--Schwarz inequality,
\begin{align}\label{cbnd}
\sigma_{jk}&\le\left(\left\langle\vc\lambda_j,\vc\lambda_j\right\rangle\left\langle\vc\lambda_k,\vc\lambda_k\right\rangle\right)^{1/2}\rho_{jk}\nonumber\\
&=\left(P_{j|j}P_{k|k}\right)^{1/2}\rho_{jk},
\end{align}
where $P$ is the stochastic matrix governing the transformation of the diagonal elements [cf. Eq.~(\ref{Pdef})]. This bound on coherence transfer in thermal operations was also derived, for all temperatures, by \'Cwikli\'nski \textit{et al.} in Ref.~\cite{SCoh}.

If the $\vc\lambda_j$'s are all pairwise linearly dependent (which is equivalent to their Gramian $Q$ being rank-$1$), then the inequality is saturated for every pair $(j,k)$. It is obvious that in such a case the ``vectors" $\vc\lambda_j$ can be chosen to be one-dimensional (i.e., scalars) and so just one diagonal Kraus operator suffices. Therefore, of all cooling maps whose associated stochastic matrix has a given diagonal, the ones with operator sum representations comprising only one diagonal Kraus operator achieve \emph{maximal coherence transfer} from the initial state to the final state. This motivates us to make the following definition:
\begin{definition}[Optimally coherent process]
A cooling map with an operator sum decomposition consisting of exactly one diagonal Kraus operator.
\end{definition}
The fact that Property~\ref{conjal} holds for such cases immediately implies
\begin{coro}\label{opco}
All optimally coherent processes are low-temperature thermal operations.
\end{coro}
\'Cwikli\'nski \textit{et al.} constructed examples of thermal processes (at general temperatures) where the bound (\ref{cbnd}) is unattainable. Our above result shows that their no-go does not hold at low temperatures, where optimal coherence transfer is always possible.

Note that every optimally coherent process achieves maximal coherence transfer \emph{given the particular diagonal elements of the associated stochastic matrix $P$}. There is an additional sense in which optimization can be achieved: We can make the diagonal elements of $P$ as large as possible. We make this idea rigorous in the following:
\begin{coro}\label{opopco}
Let two states $\rho$ and $\sigma$ satisfy:
\begin{enumerate}
\item $(\rho_{11}\dots\rho_{dd})^T\dsm(\sigma_{11}\dots\sigma_{dd})^T$;
\item The $Q$ for the pair, as defined in Theorem~\ref{coth}, exists and is positive-semidefinite and rank-$1$.
\end{enumerate}
Then,
\begin{enumerate}
\item There exists a thermal operation taking $\rho\mapsto\sigma$. Furthermore,
\item For any state $\sigma'$ such that
\be\sigma'_{jj}=\sigma_{jj}\ee
for all $j$ and $\rho\mapsto\sigma'$ is possible under cooling maps, it holds that
\be\abs{\sigma'_{jk}}\le\abs{\sigma_{jk}}\ee
for every $j\ne k$.
\end{enumerate}
In other words, for every pair $(\rho,\sigma')$ such that $\rho\mapsto\sigma'$ is possible under \emph{cooling maps}, $\rho\mapsto\sigma$ is possible under \emph{thermal operations}, where $\sigma$ has the same diagonal part as $\sigma'$ but the \emph{largest possible off-diagonal elements for the given diagonal} obtainable through cooling maps from the given initial state $\rho$.
\end{coro}
\begin{proof}
Since $\rho$ and $\sigma$ satisfy the conditions of Theorem~\ref{coth}, it follows, of course, that $\rho\mapsto\sigma$ is possible through a cooling map. In fact, since the associated $Q$ has rank $1$, Property~\ref{conjal} holds and therefore the transition is possible through a \emph{thermal operation}, proving the first assertion.

The rank-$1$ property also implies that the transition is possible by an \emph{optimally coherent process}, therefore guaranteeing optimal coherence transfer for the given diagonal part of the associated stochastic matrix $P$. However, since the $Q$ constructed in Theorem~\ref{coth} has maximal diagonal elements for the given diagonal part of the final state, so does $P$, and the second assertion follows.
\end{proof}
We saw that any optimally coherent process is a thermal operation, as is any ``coherence-killing" process. In fact, these are both special cases of a stronger result:
\begin{coro}\label{mixt}
Any mixture of optimally coherent processes can be approximated arbitrarily well by a thermal operation.
\end{coro}
\begin{proof}
We will prove that any rational convex combination of optimally coherent processes is a thermal operation. By the density of the rationals among the reals, the main claim will follow.

Let a cooling map $\cE$ be decomposable as a rational convex combination of optimally coherent processes:
\be\cE(\cdot)=\sum_{i=1}^n\fr{m_i}g\cE_i(\cdot),\ee
where $m_i$ and $g=\sum_im_i$ are positive integers and each $\cE_i$ is an optimally coherent process with Kraus operators
\begin{align}
K_i&=\left(\begin{array}{cccc}\lambda_1^{(i)}&0&\hdots&0\\
0&\lambda_2^{(i)}&0&\vdots\\
\vdots&0&\ddots&0\\
0&\hdots&0&\lambda_d^{(i)}
\end{array}\right);\nonumber\\
J_{jk}&=\mu_{jk}^{(i)}\ket j\bra k,\;j<k\in\{1\dots d\}.\nonumber
\end{align}

To realize $\cE$ as a thermal operation, we can use an ancilla $\cA$ that has a $g$-fold degenerate ground energy level $F_1$. Let $\{\ket{F_1;1}\dots\ket{F_1;g}\}$ be an orthonormal basis spanning this ground space. As we argued before, we can allow arbitrary degeneracies in the excited states of $\cA$ and take advantage of them. We use an energy-conserving unitary $U$ that satisfies
\be U\ket{E_k}\otimes\ket{F_1;t}=\lambda_k^{(i_t)}\ket{E_k}\otimes\ket{F_1;t}+\sum_{j<k}\mu_{jk}^{(i_t)}\ket{E_j}\otimes\ket{F_{jk};t},\ee
where $i_t=1$ for $t\le m_1$, $i_t=2$ for $m_1<t\le m_1+m_2$, etc. Since these states are orthogonal for different $t$'s by construction, it follows that such a unitary always exists. One may verify that the action of the resulting thermal operation on any input is identical with that of the given cooling map $\cE$.
\end{proof}
In the next supplementary note we will consider Gibbs-preserving operations, which in the low-temperature limit are defined by the constraint
\be\cE\left(\proj{E_1}\right)=\proj{E_1}.\ee
It is obvious that the set of low-temperature Gibbs-preserving operations is strictly larger than the set of cooling maps. Before moving on, let us summarize our findings on the various sets of operations that we have considered, through their inclusion hierarchy:
\begin{align}
&\left\{\textnormal{Optimally coherent processes}\right\}\nonumber\\
\subsetneq&\left\{\textnormal{Mixtures of optimally coherent processes}\right\}\nonumber\\
\subseteq&\left\{\textnormal{Low-temperature thermal operations}\right\}\nonumber\\
\subseteq&\left\{\textnormal{Cooling maps}\right\}\nonumber\\
\subsetneq&\left\{\textnormal{Low-temperature Gibbs-preserving operations}\right\}.\nonumber
\end{align}
Fig.~1 of the main text depicts a visualization of this hierarchy.

\section{Gibbs-preserving operations}
By constructing the cooling maps model we were able to get some elegant results about thermal operations. However, this reduction was made possible by the simplifying condition of low temperature. In general, when the temperature is arbitrary, thermal operations are not very yielding to elegant mathematical treatment, owing to their operational definition. In contrast, consider the following definition:
\begin{definition}[Gibbs-preserving operation]
A quantum channel $\cE$ that fixes the Gibbs state:
\be\cE(\gamma_\cS)=\gamma_\cS.\ee
\end{definition}
This definition is much more mathematically direct, and so it would seem that a model wherein the allowed processes are the Gibbs-preserving operations would lend itself better to mathematical treatment. Even if one believes that such a model is not physically motivated, and rather prefers the thermal operations model, the study of the former holds some utility. From the definition of thermal operations, it is obvious that all thermal operations are Gibbs-preserving. Therefore, by studying the Gibbs-preserving model, one could potentially gain some understanding of the more challenging thermal operations model.

Here we study the low-temperature limit of the Gibbs-preserving operations, both for its own sake and in order to see how similar the results will be to the ones we obtained from cooling maps. This will give us a sense of how close the Gibbs-preserving model might be to thermal operations at higher temperatures, where we do not yet have any mathematically amenable approximation like the cooling maps.
\subsection{The low-temperature approximation}\label{glowt}
Here the low-temperature limit is simpler to conceptualize than in the thermal operations case. We can define the lowness of temperature directly in terms of the system of interest $\cS$, instead of having to refer to the properties of the environment. If, as before, $\cS$ is a $d$-level system governed by a Hamiltonian $H_\cS$ with the non-degeneracy properties listed earlier, we can formalize the low-temperature assumption as follows:
\be k_BT\ll E_2-E_1.\ee
This leads to
\begin{equation}\label{lowtg}
\gamma_\cS\approx\proj{E_1},
\end{equation}
which will be the form in which we will use the approximation.
\subsection{Allowed operations and the canonical parametrization}
The low-temperature approximation Eq.~(\ref{lowtg}) leads to the following criterion for an evolution $\cE$ to be allowed:
\be\cE\left(\proj{E_1}\right)\approx\proj{E_1}.\ee
It is clear that the subspace spanned by $\ket{E_1}$ is treated in a privileged manner in this model. We will see this more rigorously in the upcoming sections, but in anticipation we propose the following ``canonical parametrization" of a generic state of $\cS$:
\be\rho=\left(\begin{array}{c|c}
\alpha&\vc x^\dagger\\
\hline\vc x&A\end{array}\right),\ee
where $\alpha:=\bra{E_1}\rho\ket{E_1}\ge0$ is a real scalar, $\vc x$ is a complex $(d-1)$-dimensional vector, and $A$ is a $(d-1)$-dimensional subnormalized density operator. We can identify a state with its associated set of parameters, as $\rho\equiv(\alpha,\vc x,A)$.

\subsection{The Schur complement construction}\label{schur}
The following construction will be useful in the subsequent analysis. For the present, assume for simplicity that $A$ is invertible, noting that the argument can easily be adapted to the singular case. Let
\be K_{A,\vc x}:=\left(\begin{array}{cc}1&-\vc x^\dagger A^{-1}\\0&\eins_{d-1}\end{array}\right).\ee
The map
\be\cE_{A,\vc x}:M\mapsto\cE_{A,\vc x}(M):=K_{A,\vc x}MK_{A,\vc x}^\dagger\ee
is CP. It is also invertible \footnote{Inconveniently, the map $\cE_{A,\vc x}$, while \emph{algebraically invertible}, is \emph{not functionally invertible}: its inversion requires information about $\vc x$ that is not contained in $D_\rho$ itself!}, with inverse given by the (also CP) map $\cE_{A,-\vc x}$. Its action on $\rho$ gives
\be D_\rho:=\cE_{A,\vc x}(\rho)=\left(\begin{array}{cc}\alpha-\vc x^\dagger A^{-1}\vc x&0\\0&A\end{array}\right).\ee
From the CP property of $\cE_{A,\vc x}$ and its inverse, it follows that $\rho\ge0$ is equivalent to
\begin{align}
A&\ge0,\nonumber\\
\alpha-\vc x^\dagger A^{-1}\vc x&\ge0.\nonumber
\end{align}
The quantity
\be c_\rho:=\alpha-\vc x^\dagger A^{-1}\vc x\ee
is called \emph{the Schur complement of block $A$ in the matrix $\rho$}.

In order to understand how to treat cases where $A$ is singular, note that the block $A$ in the matrix of $\rho$ can always be diagonalized by a unitary matrix of the form
\be U=\left(\begin{array}{c|c}1&0\\\hline0&V\end{array}\right),\ee
which is an allowed unitary under Gibbs-preserving operations. Since unitary operations are reversible, without loss of generality we can assume diagonal $A$ in the canonical representation
\be\rho=\left(\begin{array}{c|c}
\alpha&\vc x^\dagger\\
\hline\vc x&A\end{array}\right).\ee
If a diagonal $A$ is singular, it has some zeroes on its diagonal. But for $\rho$ to be positive-semidefinite, the components of $\vc x$ in the corresponding rows must also be zero. Therefore the quantity $\vc x^\dagger A^{-1}\vc x$ can be given a well-defined value, by considering only the terms coming from the nonzero components of $\vc x$.

\subsection{The action of allowed operations on states}
Let us characterize Gibbs-preserving operations in terms of the possible Kraus operator decompositions that they can have. If an allowed channel $\cE$ has an operator sum representation comprising the Kraus operators $\left\{K_1\dots K_r\right\}$, the requirement of fixing $\proj{E_1}$ leads to the general form
\be K_i=\left(\begin{array}{c|c}\eta_i&\vc v_i^\dagger\\
\hline0&L_i\end{array}\right).\ee
Here $\eta_i\in\bbC$, $\vc v_i\in\bbC^{d-1}$, and $L_i\in\bbC^{(d-1)\times(d-1)}$. The trace-preserving condition on $\cE$ implies that
\begin{align}
\sum_i\abs{\eta_i}^2&=1,\nonumber\\
\sum_i\eta_i\vc v_i&=\vc0,\nonumber\\
\sum_i\left(\vc v_i\vc v_i^\dagger+L_i^\dagger L_i\right)&=\eins_{d-1}.\label{krtp}
\end{align}
The action of the channel $\cE$ on a state $\rho\equiv(\alpha,\vc x,A)$ gives
\be\cE(\rho)=:\sigma\equiv(\beta,\vc y,B),\ee
where
\begin{align}
\beta&=\alpha+\sum_i\vc v_i^\dagger A\vc v_i;\nonumber\\
\vc y&=\left(\sum_i\eta_i^*L_i\right)\vc x+\sum_iL_iA\vc v_i;\nonumber\\
B&=\sum_iL_iAL_i^\dagger.\label{kaction}
\end{align}

Recall the Schur complement construction, which associates with each state $\rho$ a block-diagonal matrix $D_\rho$. Associated with the final state $\sigma$ we have $D_\sigma$. The transformation from $D_\rho$ to $D_\sigma$ can be thought of as the action of the CP map
\be\Lambda_\cE:=\cE_{B,\vc y}\circ\cE\circ\cE_{A,-\vc x}.\ee
The action of $\Lambda_\cE$ can be decomposed using the Kraus operators
\be J_i=K_{B,\vc y}K_iK_{A,-\vc x}.\ee
We find that, by virtue of the structure of the $K_i$'s, the $J_i$'s have the same form:
\be J_i=\left(\begin{array}{c|c}\eta_i&\vc u_i^\dagger\\
\hline0&L_i\end{array}\right).\ee
This leads to
\be D_\sigma=\Lambda_\cE(D_\rho)=\left(\begin{array}{cc}c_\sigma&0\\0&B\end{array}\right),\ee
where
\begin{equation}\label{schaction}
c_\sigma=c_\rho+\sum_i\vc u_i^\dagger A\vc u_i.
\end{equation}

\subsection{Monotones under Gibbs-preserving operations}\label{gmon}
Monotones are real-valued functions of the state that vary monotonically (non-increasingly or non-decreasingly) under the allowed operations. For example, in classical thermodynamics, the free energy is a monotone. Here we note a couple of monotones under Gibbs-preserving operations. By virtue of the positive-semidefiniteness of the block $A$ in the matrix of $\rho$, Eqs.~(\ref{kaction}) and (\ref{schaction}) immediately yield the conditions
\begin{align}
\beta&\ge\alpha;\nonumber\\
c_\sigma&\ge c_\rho.\label{condg}
\end{align}
These conditions lead to the following theorem, stated in the main text with a discussion of the physical significance of the quantities involved.
\begin{athm}[Theorem~\ref{thgp} of the main text]
The quantities
\be\nu_\mathrm I(\rho):=1-\alpha\ee
and
\be\nu_\mathrm C:=1-c_\rho\ee
are \emph{monotonically non-increasing} under Gibbs-preserving operations.
\end{athm}
The second monotone can be adapted to cases with singular $A$ using the line of reasoning presented at the end of the section on the Schur complement construction.

\subsection{Two-level systems and pure-state transitions}
The monotones mentioned in the previous section turn out to be sufficient in determining the feasibility of state transitions in some special cases:
\begin{prop}\label{twolg}
The conditions (\ref{condg}) are sufficient for state transitions on two-level systems, i.e., when $d=2$.
\end{prop}
\begin{proof}
In this case $A\equiv1-\alpha$ and $\vc x\equiv x$ are scalars. Similarly, among the parameters characterizing a channel $\cE$, $\vc v_i\equiv v_i$ and $L_i\equiv\lambda_i$ are now scalars. For convenience we can define the following vectors:
\begin{align}
\vc\eta&\equiv(\eta_1\dots\eta_r)^T;\nonumber\\
\vc v&\equiv(v_1\dots v_r)^T;\nonumber\\
\vc\lambda&\equiv(\lambda_1\dots\lambda_r)^T.\nonumber
\end{align}
The TP condition Eqs.~(\ref{krtp}) can now be written elegantly:
\begin{align}
\nrm{\vc\eta}&=1;\nonumber\\
\left\langle\vc\eta,\vc v\right\rangle&=0;\nonumber\\
\nrm{\vc v}^2+\nrm{\vc\lambda}^2&=1,\label{krtp2}
\end{align}
where $\langle\cdot,\cdot\rangle$ and $\nrm\cdot$ are the usual inner product and its associated geometric norm in this vector space. Under $\cE$ [cf. Eqs.~(\ref{kaction})], the component $\alpha$ transforms into
\begin{align}\label{bett}
\beta&=\alpha+\nrm{\vc v}^2(1-\alpha)\nonumber\\
&=\alpha+\left(1-\nrm{\vc\lambda}^2\right)(1-\alpha).
\end{align}
The range of values that $\beta$ can take under the conditions (\ref{condg}) is $[\alpha,1]$, and we can always choose a $\vc\lambda$ that achieves any of these values while also obeying Eq.~(\ref{krtp2}). It remains to be shown that any of the values of $c_\sigma$ allowed by (\ref{condg}) can also be achieved simultaneously.

Working out the action of the channel $\Lambda_\cE$ using the analysis that led to Eq.~(\ref{schaction}), we find that
\begin{align}
c_\sigma&=\alpha+(1-\alpha)\nrm{\vc v}^2-\fr{\abs{\left\langle\vc\lambda,\left(x\vc\eta+(1-\alpha)\vc v\right)\right\rangle}^2}{(1-\alpha)\nrm{\vc\lambda}^2}\nonumber\\
&=:f[\alpha,x,\vc v,\vc\lambda,\vc\eta].\nonumber
\end{align}
The smallest value that $c_\sigma$ can take under (\ref{condg}) is $c_\rho$. Since only the norm $\nrm{\vc\lambda}$ is relevant in achieving the requisite value of $\beta$ [See Eq.~(\ref{bett})], we are free to choose $\vc\lambda$ parallel to $\left(x\vc\eta+(1-\alpha)\vc v\right)$, so that
\be\abs{\left\langle\vc\lambda,\left(x\vc\eta+(1-\alpha)\vc v\right)\right\rangle}^2=\left(\nrm{\vc\lambda}\nrm{x\vc\eta+(1-\alpha)\vc v}\right)^2.\ee
But since $\left\langle\vc\eta,\vc v\right\rangle=0$, we have ``Pythagoras' theorem":
\be\nrm{x\vc\eta+(1-\alpha)\vc v}^2=\abs x^2\nrm{\vc\eta}^2+(1-\alpha)^2\nrm{\vc v}^2.\ee
This, combined with Eq.~(\ref{krtp2}), gives us
\be f[\alpha,x,\vc v,\vc\lambda,\vc\eta]=c_\rho.\ee
This shows that the least possible value of $c_\sigma$ can be achieved. The largest possible value of $c_\sigma$ is $\beta$. This can be achieved by choosing $\vc\lambda$ to be \emph{orthogonal} to $\left(x\vc\eta+(1-\alpha)\vc v\right)$, again without affecting the ability to achieve the desired $\beta$.

To achieve any intermediate value of $c_\sigma$, we can choose $\vc\lambda$ to have an intermediate direction.
\end{proof}
\begin{coro}\label{gpure}
The conditions (\ref{condg}) are sufficient when $\rho$ and $\sigma$ are both pure.
\end{coro}
\begin{proof}
Consider a pure state
\be\ket\psi=t\ket{E_1}+\left(1-\abs t^2\right)^{1/2}\ket\phi,\ee
where $\ket\phi$ is a normalized vector such that $\bra\phi E_1\rangle=0$.
Using a unitary operation of the form
\be U=\left(\begin{array}{c|c}1&0\\\hline0&V\end{array}\right),\ee
which is allowed under Gibbs-preserving operations, we can always reversibly transform $\ket\psi$ to a state of the form
\be\ket{\tilde\psi}=t\ket{E_1}+\left(1-\abs t^2\right)^{1/2}\ket{E_2}.\ee
Therefore, every state transition question involving a pair of $d$-dimensional pure states can be reduced to one involving pure states in the $2$-dimensional subspace spanned by $\{\ket{E_1},\ket{E_2}\}$. By Supplementary Proposition~\ref{twolg}, the claim follows.
\end{proof}

\end{document}